\def\draft{0}
\documentclass[11pt]{article}

\usepackage{amsfonts}
\usepackage{amsmath}
\usepackage{amssymb}
\usepackage{amsthm}
\usepackage{mathabx}
\usepackage{thmtools, thm-restate}
\usepackage{bbm}
\usepackage{bm}
\usepackage{dsfont}
\usepackage{fullpage}
\usepackage{tikz}
\usepackage[most]{tcolorbox}
\usepackage{xcolor}
\usepackage{ytableau}

\newcommand*{\myfont}{\fontfamily{bch}\selectfont}
\DeclareTextFontCommand{\textmyfont}{\myfont}

\usepackage[linesnumbered,ruled,vlined]{algorithm2e} 
\SetKwComment{Comment}{/* }{ */}

\SetCommentSty{mycommfont}

\usepackage[hidelinks]{hyperref}
\usepackage[nameinlink]{cleveref}
\usepackage{url}            
\usepackage{footnotebackref}

\newtheorem{theorem}{Theorem}[section]

\newtheorem{lemma}[theorem]{Lemma}

\newtheorem{definition}[theorem]{Definition}

\newtheorem{remark}[theorem]{Remark}

\newcommand{\poly}{\mathop{\mathrm{poly}}}

\newcommand{\F}{\mathbb{F}}

\newcommand{\N}{\mathbb{N}}

\newcommand{\Boo}{\{0,1 \}}

\newcommand{\bigO}{\mathcal{O}}

\newcommand{\IPS}{\mathsf{IPS}}
\newcommand{\IPSLIN}{\mathsf{IPS}_{\mathsf{LIN}}}
\newcommand{\IPSLINp}{\mathsf{IPS}_{\mathsf{LIN}'}}

\newcommand{\ml}{\mathsf{ml}}
\newcommand{\NCo}{\mathsf{NC}^1}
\newcommand{\NCt}{\mathsf{NC}^2}
\newcommand{\nCtwo}{{[n] \choose 2}}

\newcommand{\paren}[1]{\left( #1 \right)}

\newcommand{\set}[1]{\left\{ #1 \right\}}

\newcommand{\x}{\mathbf{x}}
\newcommand{\y}{\mathbf{y}}
\newcommand{\z}{\mathbf{z}}
\newcommand{\ubar}{\mathbf{u}}
\newcommand{\vbar}{\mathbf{v}}
\newcommand{\wbar}{\mathbf{w}}

\definecolor{thmcolor}{RGB}{235, 235, 235}
\definecolor{citecolor}{RGB}{1, 210, 56}
\definecolor{lemmacolor}{RGB}{130, 169, 252}

\newtcolorbox{algobox}{colback=lightgray!5!white,colframe=lightgray!75!black}
\newtcolorbox{thmbox}{colback=thmcolor!5!white,colframe=black!75!black}
\newtcolorbox{lemmabox}{colback=lemmacolor!5!white,colframe=blue!75!blue}

\usepackage{setspace}
\usepackage{enumerate}
\usepackage[parfill]{parskip}
\usepackage{changepage}
\usepackage{framed}

\allowdisplaybreaks
\hypersetup{
	colorlinks,
	linkcolor={blue},
	citecolor={citecolor},
	urlcolor={blue}
}

\usepackage[
backend=biber,
style=alphabetic,
sorting=nyt,
backref=true,
maxcitenames = 8,
mincitenames = 5,
maxalphanames = 8,
minalphanames = 5,
maxnames = 10,
minnames = 5
]{biblatex}

\ifnum\draft=1
\newcommand{\snote}[1]{{\color{teal} [Srikanth: #1]}}
\newcommand{\nnote}[1]{{\color{pink} [Nutan: #1]}}
\newcommand{\vnote}[1]{{\color{blue} [Varun: #1]}}
\newcommand{\mnote}[1]{{\color{purple} [Magnus: #1]}}
\newcommand{\anote}[1]{{\color{brown} [Amik: #1]}}
\else
\newcommand{\snote}[1]{}
\newcommand{\nnote}[1]{}
\newcommand{\vnote}[1]{}
\newcommand{\mnote}[1]{}
\newcommand{\anote}[1]{}
\fi

\def\anon{0}
\date{\today}

\begin{document} 
\title{Separation Results for Constant-Depth and Multilinear Ideal Proof Systems}

    \if\anon1{}\else{    
    \author{Amik Raj Behera\thanks{Department of Computer Science, University of Copenhagen, Denmark. Supported by Srikanth Srinivasan's start-up grant from the University of Copenhagen. Email: \texttt{ambe@di.ku.dk} } \and Magnus Rahbek Dalgaard Hansen \thanks{IT University of Copenhagen, Denmark, \texttt{Email: ramh@itu.dk} Supported by the Basic Algorithms Research Copenhagen (BARC), funded by VILLUM Foundation Grant 54451.} \and Nutan Limaye \thanks{IT University of Copenhagen, Denmark, \texttt{Email: nuli@itu.dk}. Supported by Independent Research Fund Denmark (grant agreement No. 10.46540/3103-00116B) and is also supported by the Basic Algorithms Research Copenhagen (BARC), funded by VILLUM Foundation Grant 54451.} \and
		Srikanth Srinivasan \thanks{Department of Computer Science, University of Copenhagen, Denmark. Supported by the European Research Council (ERC) under grant agreement no. 101125652 (ALBA). Email: \texttt{srsr@di.ku.dk} }}
     }\fi

	\maketitle
        \pagenumbering{arabic}

\begin{abstract}
In this work, we establish separation theorems for several subsystems of the Ideal Proof System ($\mathsf{IPS}$), an algebraic proof system introduced by Grochow and Pitassi (J. ACM, 2018). Separation theorems are well-studied in the context of classical complexity theory, Boolean circuit complexity, and algebraic complexity. 

In an important work of Forbes, Shpilka, Tzameret, and Wigderson (Theory of Computing, 2021), two proof techniques were introduced to prove lower bounds for subsystems of the $\mathsf{IPS}$, namely the \emph{functional method} and the \emph{multiples method}. We use these techniques and obtain the following results.\\ 

\begin{enumerate}
    \item \textbf{Hierarchy theorem for constant-depth $\mathsf{IPS}$.} Recently, Limaye, Srinivasan, and Tavenas (J. ACM 2025) proved a hierarchy theorem for constant-depth algebraic circuits. We adapt the result and prove a hierarchy theorem for constant-depth $\mathsf{IPS}$. We show that there is an unsatisfiable multilinear instance refutable by a depth-$\Delta$ $\mathsf{IPS}$ such that any depth-($\Delta/10)$ $\mathsf{IPS}$ refutation for it must have superpolynomial size. This result is proved by building on the \emph{multiples method}.\\
    
    \item \textbf{Separation theorems for multilinear $\mathsf{IPS}$.} In an influential work, Raz (Theory of Computing, 2006)  unconditionally separated two algebraic complexity classes, namely multilinear $\mathsf{NC}^{1}$ from multilinear $\mathsf{NC}^{2}$. 
    In this work, we prove a similar result for a well-studied fragment of multilinear-$\mathsf{IPS}$.

    \indent Specifically, we present an unsatisfiable instance such that its \emph{functional refutation}, i.e., the unique multilinear polynomial agreeing with the inverse of the polynomial over the Boolean cube, has a small multilinear-$\mathsf{NC}^{2}$ circuit. However, any multilinear-$\mathsf{NC}^{1}$ $\mathsf{IPS}$ refutation ($\mathsf{IPS}_{\mathsf{LIN}}$) for it must have superpolynomial size. 
    This result is proved by building on the \emph{functional method}.\\
\end{enumerate}

Given a polynomial $p(\mathbf{x})$, let $\textsf{Image}(p(\mathbf{x}))$ denote the set of values obtained when $p(\mathbf{x})$ is evaluated over the Boolean cube. Our crucial observation is that if the cardinality of this set is $\mathcal{O}(1)$, then the functional method and multiples method can be used to prove separation theorems for subsystems of the $\mathsf{IPS}$. We obtain such polynomial instances by lifting the hard instances arising from algebraic circuit complexity with \emph{addressing gadgets}. 
\end{abstract} 

\paragraph{Acknowledgments.} The authors would like to thank Varun Ramanathan for helpful discussions during the early stages of the project. 

\thispagestyle{empty}

\tableofcontents
\thispagestyle{empty}

        \newpage
        
\setcounter{page}{1}

\section{Introduction}
A proof system is defined by a collection of axioms together with a set of inference rules that determine how new statements can be derived from existing ones. 
The objective is to begin with the given axioms and apply these inference rules to derive theorems (or tautologies) within the system. 
A proof system is said to be \emph{sound} if it proves only valid statements, and \emph{complete} if every valid statement can be derived within it. 

The field of \emph{Propositional Proof Complexity} studies the comparative strength and efficiency of such systems in the propositional setting. 
A foundational result by Cook and Reckhow~\cite{CR79} established that if one could exhibit propositional tautologies that require exponentially large proofs (that is, proofs whose length---roughly corresponding to the number of inference steps---grows exponentially) in \emph{every} propositional proof system, then this would separate the complexity classes $\mathsf{NP}$ and $\mathsf{coNP}$. 
Thus, lower bounds in proof complexity are deeply connected to some of the central open problems in computational complexity theory. 

In this work, we focus on \emph{algebraic proof systems}, in which we consider unsatisfiable systems of polynomial equations, and reasoning proceeds through algebraic manipulations such as addition and multiplication of polynomials. Here, more specifically, we consider an algebraic proof system called the Ideal Proof System ($\IPS$), which was introduced by Grochow and Pitassi~\cite{GP14}. In the last decade, different facets of this proof system have been investigated by a series of works~\cite{FSTW21,AF22,GHT,ST23,HLT24,LuST25,BLRS25,EGLT25,CGMS25}. Our paper contributes to this line of research by studying separation theorems in this context.

Hierarchy theorems are a class of separation theorems that establish that more resources yield strictly more power. For instance, the classical Time Hierarchy Theorem~\cite{Ladner} states that increasing the available running time strictly increases the computational power of a machine. Analogous results are known for several resources such as space~\cite{sp-hierarchy}, circuit depth \cite{sipser83,
hasted86}, and circuit size \cite{Jukna12,shannon49}.  
Here, we raise the question about hierarchy theorems, and more generally we study separation theorems for different subsystems of the $\IPS$. 

In order to describe our results, we first start by giving a brief introduction to the Ideal Proof Systems~\Cref{sec:intro-ips}. We then review some results from Algebraic Circuit Complexity in~\Cref{sec:intro-algcomp}, which we will use crucially in our work. Our results and techniques are described in~\Cref{sec:intro-const-depth-results} and in~\Cref{sec:intro-mult-separation}.


\subsection{Ideal Proof System}
\label{sec:intro-ips}
We begin by recalling the general framework of algebraic proof systems, focusing on the so-called \emph{static} systems.\footnote{In the literature, systems of this type are often referred to as static proof systems. Other variants, where proofs are given line-by-line, are known as dynamic proof systems. In this paper, we only consider static systems.} 
Let $\x$ denote the set of variables $\{x_1, x_2, \ldots, x_N\}$. 
Given a collection of polynomial axioms $f_1(\x), f_2(\x), \ldots, f_m(\x) \in \mathbb{F}[\x]$, the goal is to certify that there is \emph{no} Boolean assignment to the variables that simultaneously satisfies all the equalities 
$f_1(\x) = f_2(\x) = \cdots = f_m(\x) = 0$. 
To ensure that solutions are Boolean, the system is augmented with the \emph{Boolean axioms} $\{x_i^2 - x_i = 0\}_{i \in [n]}$. 

By Hilbert’s Nullstellensatz, the unsatisfiability of this augmented system can be expressed algebraically. 
Specifically, if the system has no common zero over $\mathbb{F}$, then there exist polynomials $A_1(\x), \ldots, A_m(\x)$ and $B_1(\x), \ldots, B_N(\x)$ such that
\begin{equation}
    \sum_{i \in [m]} A_i(\x) \cdot f_i(\x) + 
    \sum_{j \in [N]} B_j(\x) \cdot (x_j^2 - x_j) = 1.
\end{equation}
This identity serves as a \emph{refutation} (or \emph{proof}) of the original system. 
The complexity of such a refutation is measured in terms of the complexity of the polynomials $\{A_i\}$ and $\{B_j\}$. 

In the \emph{Ideal Proof System ($\IPS$)} introduced by Grochow and Pitassi~\cite{GP14}, the polynomials $A_i(\x)$ and $B_j(\x)$ are represented by algebraic circuits. 
This gives rise to natural complexity parameters such as the \emph{circuit size} and \emph{circuit depth} of $\IPS$ proofs. We now formally define the ideal proof system.\\ 

\begin{definition}[Ideal Proof System~\cite{GP14}]  
\label{def:ips}
Let $f_1,\ldots, f_m \in \F[x_1,\ldots, x_n]$ be a system of unsatisfiable polynomials over the Boolean cube $\Boo^{n}$. In other words, there is no Boolean assignment $\mathbf{a} \in \{0,1\}^n$ to the variables $x_1,\ldots, x_n$ so that $f_i(\mathbf{a}) = 0$ for all $i\in [m].$

Given a class of algebraic circuits $\mathcal{C}$, a \emph{$\mathcal{C}$-$\IPS$ refutation} of the system of equations defined by $f_1,\ldots,f_m$ is an algebraic circuit $C\in \mathcal{C}$ in variables $x_1,\ldots,x_n,y_1,\ldots, y_m, z_1,\ldots, z_n$ such that
\begin{itemize}
    \item $C(\mathbf{x},\mathbf{0},\mathbf{0}) = 0$, and
    \item $C(\mathbf{x},f_1,\ldots,f_m,x_1^2-x_1,\ldots,x_n^2-x_n) = 1.$
\end{itemize}
The size of the refutation is the size of the circuit $C.$

Further, if the circuit $C$ has individual degree at most $1$ in the variables $\mathbf{y}$ and $\mathbf{z}$, then we say that $C$ is a \emph{$\mathcal{C}$-$\IPSLIN$ refutation}. If the circuit $C$ has individual degree at most $1$ in the variables $\mathbf{y}$ (but not necessarily in $\mathbf{z}$), then $C$ is said to be a \emph{$\mathcal{C}$-$\IPSLINp$ refutation.\\}
\end{definition}

The general $\IPS$ where the class $\mathcal{C}$ is allowed be to be an algebraic circuit can polynomially simulate Extended Frege~\cite{GP14}, one of the strongest known propositional proof systems. Moreover, establishing lower bounds for these kind of general $\IPS$ would imply strong algebraic circuit lower bounds, a central open problem in algebraic complexity. 

While this continues to be an ambitious open problem we have many compelling new lower bound results for several restricted classes $\mathcal{C}$ such as roABPs, constant-depth circuits, and multilinear formulas~\cite{FSTW21,GHT,HLT,BLRS25,EGLT25}.

These lower bounds were established by using the already known lower bounds for the corresponding models of computation in algebraic complexity. We are also inspired by this framework. Namely, we use the separation results and hierarchy theorems from algebraic complexity theory to obtain similar results for the $\IPS$. We now review the known separation results. 

\subsection{Algebraic Circuit Complexity}
\label{sec:intro-algcomp}
We start by recalling some of the standard models of computation relevant to our results.

\textit{Algebraic circuits, formulas, constant-depth circuits, multilinear polynomials and circuits.} An \emph{algebraic circuit} is a directed acyclic graph in which each node either computes a sum (or a linear combination) of its inputs, or a product of its inputs. The leaf nodes are either variables or constants. The size of an algebraic circuit is the number of edges or wires in the circuit, and the depth of an algebraic circuit is the longest path from a leaf node (a source) to the output node (a sink). An \emph{algebraic formula} is an algebraic circuit where the output of each node feeds into at most one other node; in other words, the underlying graph of an algebraic formula is a tree. An algebraic circuit/formula is said to be constant-depth circuit/formula, if its depth is a fixed constant independent of other parameters. 

A polynomial $f(\mathbf{x}) \in \F[x_{1},\ldots,x_{n}]$ is \emph{multilinear} if in every monomial of the polynomial, the degree of any variable is at most $1$. 
An algebraic circuit/formula is multilinear if every gate computes a multilinear polynomial. An algebraic circuit is syntactically multilinear if polynomials computed by the children of any multiplication gate compute polynomials on disjoint sets of variables. 

\subsubsection{Separation results}
\label{sec:sep-results}

Our work relies heavily on the separation results known in algebraic complexity theory. Our result related to multilinear $\IPS$ is based on the following multilinear separation result. 

\paragraph{Multilinear formulas vs circuits.} One of the celebrated results in algebraic complexity is the separation between multilinear formulas and multilinear circuits. The result was established in an influential work of Raz~\cite{raz-ml-formulavscircuit}, which presented a polynomial that is computed by a polynomial sized multilinear circuit, but any multilinear formula for it requires superpolynomial size. 

The key idea involves coming up with a complexity measure for polynomials, which attains a large value for the hard polynomial, but it is considerably small for all multilinear formulas of small size. The measure from \cite{raz-ml-formulavscircuit} is defined as follows.\\

\begin{definition}[Rank measure~\cite{raz-ml-formulavscircuit}]
    \label{def:raz-measure}
    Let $\x = \{x_1, \ldots, x_{2n}\}$. Let $\y \cup \z$ be an equipartition of $\x$, i.e. $|\mathbf{y}| = |\mathbf{z}|$. For a given polynomial $f(\x)$, let $M_{\y,\z}(f)$ be a matrix with rows labeled by multilinear monomials in $\y$ variables and columns labeled by multilinear monomials in $\z$ variables. For a monomial $m_{\y}$ in $\y$ variables and $m_{\z}$ in $\z$ variables, the $M_{\y,\z}(f)[m_{\y}, m_{\z}]^{th}$ entry of the matrix is the coefficient of the monomial $m_{\y}\cdot m_{\z}$ in $f$. The measure is the rank of this matrix. 
    
    We will say that a polynomial $f$ is full-rank with respect to a partition $\y, \z$ if the rank of $M_{\y,\z}(f)$ is full, i.e. $2^n$.\\ 
\end{definition}

It was shown by Raz~\cite{raz-ml-formulavscircuit} that a multilinear formula computing any full-rank polynomial $f(\x)$ requires size $n^{\Omega(\log n)}$. In our work, we build on the full-rank polynomial defined in a subsequent work of Raz and Yehudayoff~\cite{RY08}.

\paragraph{Constant-depth hierarchy theorem}
In our work we establish a constant-depth hierarchy theorem for constant-depth $\IPS$. For this, the starting point is the constant-depth hierarchy theorem by Limaye, Srinivasan, and Tavenas~\cite{LST}. For every depth $\Delta$, they design a polynomial that is computable by polynomial size depth $\Delta$ circuits but any circuit of depth even one smaller than $\Delta$ requires superpolynomial size for it. As circuits can be converted to formulas with polynomial blow-up when the depth is constant, we state the formula version of the hierarchy theorem below. Formally, it states the following.\\ 

\begin{theorem}[Constant-depth algebraic formulas hierarchy]\label{thm:lst-hierarchy}
\cite[Theorem 5]{LST}. For every depth parameter $\Delta = \bigO(1)$, there exists an explicit set-multilinear polynomial $Q_{\Delta} \in \F[x_{1},\ldots,x_{n}]$ such that:
\begin{enumerate}
    \item There exists a constant-free\footnote{A circuit or formula $C(\mathbf{x})$ is constant free if it has no constants except at the inputs where all input gates have labels from $\mathbf{x} \cup \{-1,0,1\}$.} algebraic formula with input gates carrying labels from $\mathbf{x} \cup \{0,1\}$ which computes $Q_{\Delta}(\mathbf{x})$ in depth $\Delta$ and size $s$.
    \item Any algebraic formula computing $Q_{\Delta}(\mathbf{x})$ in depth $(\Delta/2-1)$ requires size $s^{\omega(1)}$.
\end{enumerate}
\end{theorem}

\begin{remark}
    In~\cite{LST}, a tighter separation is obtained. Namely, the depth hierarchy separates two consecutive depths, $\Delta$ vs. $\Delta-1$. However, the formulas arising from this are not constant-free. Depth hierarchy for constant-free formulas can be obtained by a slight loss in parameters, as mentioned in the statement above. 
\end{remark}

We are now ready to state our contributions.

\subsection{Results and Techniques: the constant-depth $\IPS$ hierarchy theorem}
\label{sec:intro-const-depth-results}

As our first result, we prove a depth-hierarchy theorem for constant-depth $\IPS$. More specifically, we prove the following theorem.\\ 

\begin{theorem}[Constant-depth IPS hierarchy]\label{thm:main-constant-depth}
Let $\mathbb{F}$ be a field of characteristic zero.
The following holds for every growing parameter $N \in \mathbb{N}$ and a depth parameter $\Gamma \in \mathbb{N}$ where $\Gamma = \bigO(1)$. 
For every depth parameter $\Gamma$, there exists a multilinear polynomial \anote{do we need to state the degree in the bound?} $f_{\Gamma} \in \F[x_{1},\ldots,x_{N}]$ which is unsatisfiable over $\Boo^{N}$ (i.e. there exists no $\mathbf{a} \in \Boo^{N}$ for which $f_\Gamma(\mathbf{a}) = 0$) such that the following two conditions hold:
\begin{enumerate}
    \item There exists an $\IPS$ refutation for $f_{\Gamma}(\mathbf{x})$ in depth $\Gamma$ and size $\bigO(s^{5})$.
    \item Any $\IPS$ refutation for $f_{\Gamma}(\mathbf{x})$ with depth $\leq \Gamma/10$ requires size $s^{\omega(1)}$. 
\end{enumerate}
\end{theorem}

To describe the proof strategy, we prove the above theorem for a simpler case. Let $Q_\Delta(\x)$ be the polynomial used by~\cite{LST} in~\Cref{thm:lst-hierarchy}. Now consider $g_\Delta(\x, y)$ defined as $Q_\Delta(\x) \cdot y \cdot (1-y)$, where $y$ is a new variable. First, observe that the polynomial evaluates to $0$ over the Boolean cube. In fact, this is true if we only consider the Boolean evaluations of $y$. Therefore, $g_\Delta(\x,y) - 1$ is unsatisfiable. 
Moreover, $g_\Delta(\x,y) - 1 \equiv 1 \mod{ y^2 -y}$. And it is easy to see that $g_\Delta(\x,y)$ has the same upper bound as $Q_\Delta(\x)$. Thus, we get the upper bound. 

For the lower bound, we will use the multiples method. The method was introduced in~\cite{FSTW21} and it has been used successfully for $\IPS$ lower bounds in~\cite{FSTW21,AF22,Andrews2025}. We now describe how one can use this method to obtain a lower bound.

Consider $g_{\Delta}(\x,y)-1$. Using~\Cref{thm:lst-hierarchy}, we know that it does not have polynomial sized formulas of depth $(\Delta/2-1)$. Moreover, due to the recent work on factors of constant-depth formulas~\cite{BKRRSS}, we also know that every multiple of the polynomial of depth $(\Delta/2-\bigO(1))$ must have superpolynomial size. That is, the polynomial $g_{\Delta}(\mathbf{x},y) - 1$ and all its multiples are hard for depth $(\Delta/2 - \bigO(1))$. This property suffices for the multiples method to be applicable, as we explain next. Specifically, our polynomial system is $f_\Delta = g_{\Delta}(\x,y)-1$, $\{x_i^2-x_i\}_{i \in [N]}$, and $\{y^2-y\}$. IPS refutation is such that 
$C(\x, y,u, \mathbf{0}, 0) = 0$ and $C(\x,y, f_{\Delta}, \x^2 - \x, y^2 -y) =1$,
where $\x^2-\x$ denotes $\{x_i^2-x_i\}_{i}$. 

We now express $C(\x,y, f_{\Delta}, \x^2 - \x, y^2 -y)$ 
as a univariate in $f_\Delta$ and we obtain \[\sum_{i \geq 1 } C_i(\x, y, \x^2-\x, y^2-y) f_\Delta^i = 1 - C(\x, y, 0, \x^2-\x, y^2-y)\]

for some $C_is$. This shows that a multiple of $f_\Delta$ has the same complexity as $C(\x,y, 0,\x^2-\x, y^2-y)$. But we know that all the multiples of $f_\Delta$ are hard. This gives an $\IPS$ lower bound.

\begin{remark}
    \label{rmk:non-mult}
    The above proof outline basically proves~\Cref{thm:main-constant-depth} when the hard instance is non-multilinear. Note that the polynomial $Q_\Delta$ is multilinear, but the hard instance is non-multilinear in $y$. We extend the ideas presented in the outline above and obtain a hard instance that is multilinear. 
\end{remark}

\paragraph{A hard multilinear instance.} The technical challenge in our work is designing an instance that is \emph{multilinear}. We seek such an instance for the following reason: In algebraic proof complexity, the goal is to find an instance that is itself quite \emph{easy} to compute, but its refutation is hard. There are many different ways to quantify easiness. However, one of the standard ways is to ask for a hard instance to be multilinear. Almost all the known hard instances in this literature are multilinear. (See for instance~\cite{FSTW21,GHT,HLT}). In fact, the hard instance naturally arising from the algebraic encoding of CNF SAT is also multilinear. 

There are some challenges that arise when we require a multilinear instance. 

    \textbf{Making the instance multilinear.} As mentioned above, the hard polynomial $Q_\Delta$ from~\cite{LST} is already multilinear. However, unfortunately, we do not know how to upper bound the complexity of IPS refutations of $Q_\Delta$ itself with respect to depth $\Delta$. To fix this, we modify $Q_\Delta$ so that the (new) hard instance takes only 
    $\bigO(1)$ distinct values over Boolean evaluations. 
    We rely on this fact for our upper bound (see the discussion below). Our upper bound proof is further simplified if the polynomial takes only $0$-$1$ values over the Boolean cube. So, we bake these two properties into the design of the polynomial: (a) it is multilinear, and (b) it takes only Boolean values over the Boolean hypercube. 

    We start with $Q_\Delta$ as in~\Cref{thm:lst-hierarchy} and take the constant-depth formula implementation for it. In this formula, we introduce \emph{an addressing gadget} for each $+$ gate. An addressing gadget is a multilinear polynomial that works like a \emph{multiplexer}. For a $0$-$1$ values as inputs to the gadget polynomial, it \emph{activates} one of the inputs to the plus gate and \emph{suppresses} all the other inputs. As a result, if we inductively maintain $0$-$1$ evaluations for all the gates over the Boolean cube, the gadget allows us to propagate this property to the next gate. 
    
    \textbf{Proving the upper bound.} 
    We work with the formula $C$ that \emph{computes} our hard polynomial instance. By construction of the polynomial, we have the guarantee that every gate in the formula evaluates to $0$ or $1$ over the hypercube. Using this fact, we prove by induction on the depth of the formula that for a gate $g$, the polynomial $g^2 - g$ is in the ideal generated by $\{g_i^2-g_i\}_i$ and the Boolean axioms, where $\{g_1, \ldots, g_t\}$ are the inputs to the gate $g$. This suffices to obtain the overall upper bound.

    \textbf{Proving the lower bound.} The lower bound proof proceeds by observing that there exists an assignment to  the gadget variables such that under that assignment, the hard instance becomes equal to the hard polynomial $Q_\Delta$ from~\Cref{thm:lst-hierarchy}. As this polynomial and all its multiples are hard (due to~\cite{BKRRSS}), we obtain our lower bound using the multiples method.

    The advantage of using the multiples method (instead of the functional) is that we obtain the lower bound theorem for $\IPS$ and not just for the more restrictive $\IPSLIN$.



\subsection{Results and Techniques: multilinear-$\NCo$ vs. multilinear-$\NCt$}
\label{sec:intro-mult-separation}    

In this section, we state our separation theorem for a multilinear-$\IPS$ system. We say that for a polynomial instance $Q(\x)$ unsatisfiable over the Boolean cube, a functional refutation is a polynomial $G(\x)$ such that $G(\x) \cdot Q(\x) \equiv 1 \mod{\x^2 -\x}$. Further, we will say that it is a multilinear functional refutation is $G(\x)$ is multilinear. 
We prove the following theorem. 

\begin{restatable}[multilinear $\NCo$ vs multilinear $\NCt$-$\IPS$]{theorem}{mainthmmultNC}\label{thm:main-mnc1-mnc2}
     Fix a field, $\F$ of characteristic $0$. For every growing parameter $N \in \N$, there is a multilinear polynomial $Q \in \F[x_1, \ldots, x_N]$ which is unsatisfiable over $\{0,1\}^N$ such that 
    \begin{enumerate}
        \item There is a multilinear functional refutation for $Q(\mathbf{x})$, say $G(\x)$, computable by a syntactic multilinear circuit of polynomial size and $O(\log^2 N)$ depth. 
        \item Any multilinear-$\NCo$-$\IPSLINp$ for it requires size $N^{\Omega(\log N)}$. 
    \end{enumerate}
\end{restatable}


\begin{remark}
    \label{rmk:caveat}
    Note that in the above theorem the lower bound holds for multilinear-$\NCo$-$\IPSLINp$. However, the upper bound is not multilinear-$\NCt$-$\IPSLINp$. Instead, we only get that the refutation has multilinear-$\NCt$ circuits modulo the Boolean axioms. We do not get a bound on the complexity of the refutations of the Boolean axioms. In spite of this, we believe that the above result gives something we did not know before. 

    \begin{itemize}
        \item Lower bounds for multilinear-$\NCo$-$\IPSLINp$ are known since the work of~\cite{FSTW21}. Their hard instance is a \emph{lifted subset-sum}, i.e. $f(\x,\z) = \sum_{i,j} z_{i,j} x_i x_j$. We observe that its functional refutation is quite hard. Specifically, it encodes the Clique polynomial over the Boolean cube. This means that it cannot have small functional refutations unless VP equals VNP. (See~\Cref{sec:appendix}). 
        \item It is known that there are subsystems of multilinear-$\IPSLIN$ and multilinear-$\IPSLINp$ that can refute interesting unsatisfiable instances (Section 4,~\cite{FSTW21}). For example, they can refute the subset-sum  instances of the type $\sum_{i} \alpha_i x_i - \beta$, where $\alpha_i$s are $O(1)$ and $\beta$ is chosen such that the instance becomes unsatisfiable\footnote{They can allow slightly general $\alpha_i$s. See Proposition 4.15 from~\cite{FSTW21} for more details.}. While such instances have multilinear upper bounds, the upper bound proofs seem to heavily rely on the fact that the subset-sum polynomial has degree $1$. Consider a simple instance $xy = 2$. This is an unsatisfiable instance over the Boolean cube. Here is one of its refutations: \[x^2y^2 - xy = x^2y^2 - x^2y + x^2y -xy = x^2(y^2-y) + y(x^2-x).\]
        Notice that the refutation for the Boolean axiom $y^2-y$ is not multilinear. (There is another refutation for the same and in that, the refutation of the other Boolean axiom is not multilinear.) In fact, any refutation of this example is not multilinear. (See Example 4.7 in~\cite{FSTW21}.) This gives an indication that a degree-$2$ (or degree-greater-than-$2$) hard instance may not necessarily have multilinear proofs. 
        \item Our lower bounds are obtained using the functional method. This ensures that the hardness of our instance can be ascribed to the hardness of refuting the instance irrespective of the complexity of the refutations of the Boolean axioms. Thus, the result achieves a separation for the functional refutation of our instance. 
    \end{itemize}
\end{remark}

To describe the components of the proofs, we start with a very simple example. Let $z \in\{0,1\}$. In this case, it is easy to see that $2-z$ is unsatisfiable and $2-z \times \frac{1+z}{2} \equiv 1$ modulo $z^2-z$. That is, $\frac{1+z}{2}$ is a refutation of $2-z$ modulo the Boolean axioms. We make use of this idea in our proof. In order to prove the theorem, we again design a polynomial $p(\x)$ such that it evaluates to $0$ or $1$ over the Boolean cube. Then, our unsatisfiable instance is $2-p(\x)$ and its functional refutation is $(1+p(\x))/2$. 

If we can design a multilinear polynomial such that 
\begin{itemize}
    \item it is computed by multilinear $\NCt$ circuits, 
    \item any multilinear $\NCo$ circuit for it requires $N^{\Omega(\log N)}$ size,
    \item and it takes only Boolean values over the Boolean cube
\end{itemize}
then we get the separation. From a famous work of Raz~\cite{raz-ml-formulavscircuit} we get a polynomial that satisfies the first two properties listed above. A subsequent work of Raz and Yehudayoff~\cite{RY08} also gives another candidate polynomial. Unfortunately, neither of them have the third property. We tweak the polynomial from~\cite{RY08} using the addressing gadgets to ensure that we get a multilinear polynomial with all these properties. 

\paragraph{Applicability of the technique.} The proof method used for proving~\Cref{thm:main-mnc1-mnc2} points to its applicability to other scenarios. For example, the same proof method can be applied in the context of the constant-depth hierarchy theorem (as in~\Cref{thm:main-constant-depth}). The upper bound stays as is, but the lower bound is obtained using the method described above. This will work and will give a lower bound for $\IPSLIN$ instead of a lower bound for $\IPS$. 

There are other separation results known in algebraic complexity, especially in the multilinear setting. For example, results of~\cite{RY09,CELS-18-multilinear-hierarchy}. Our proof method is likely to be applicable 
in all these settings to obtain separation results originating for different sub-systems of the $\IPS$ from these separation results, just like we proved~\Cref{thm:main-mnc1-mnc2} from the separation results of Raz~\cite{raz-ml-formulavscircuit,RY08}.

\section{Constant-depth Hierarchy}
In this section, we will prove \Cref{thm:main-constant-depth}. To do so, we start with the hard polynomials from \Cref{thm:lst-hierarchy} and modify them by using addressing gadgets.

Throughout this section, we will assume without loss of generality that every algebraic formula $C$ is layered and has alternating addition and multiplication gates, with the top gate being an addition gate. For every gate $g$ in a formula $C(x_{1},\ldots,x_{N})$,
\begin{itemize}
    \item $f_g(x_{1},\ldots,x_{N})$ will denote the polynomial computed at the gate $g$.
    \item Let $\mathrm{depth}(g)$ denote the depth of gate $g$, i.e. the length of the longest path from inputs to the gate. Let $\mathrm{size}(g)$ denote the number of wires in the sub-formula rooted at $g$. Finally,  let $\mathrm{fanin}(g)$ denote the fan-in of gate $g$. 
\end{itemize}

\subsection{Adding addressing gadgets at sum gates}
In this subsection, we define a modification for any given algebraic formula, ensuring that the new formula is a $\set{0,1}$-valued on Boolean inputs. Furthermore, the polynomial computed by the original formula can be easily retrieved from the new formula via partial evaluation of its variables.\\

\begin{definition}\label{def:addressing_fns}
    Let $n \in \mathbb{N}$. For each $0 \leq j \leq n$, let $t_n \in \mathbb{N}$ denote the smallest $t_n$ such that $2^{t_n} >n$. Let $B_{n,0}(j) \subseteq \{0,1,...,t_n\}$ denote the indices which are $0$ in the binary representation of $j+2^{t_n}$. Similarly, let $B_{n,1}(j) \subseteq \{0,1,...,t_n\}$ denote the set of indices which are 1 in the binary representation of $j+2^{t_n}$.

    The \emph{addressing gadget} of $j$ in $n$ is defined as
    $$A_{n,j}(y_0,...,y_{t_n}) = \prod_{i \in B_{n,0}(j)}(1-y_i)\prod_{i \in B_{n,1}(j)}y_i.$$\\
\end{definition}

Note that $A_{n,j}$ uses the same set of variables $\{y_0,...,y_{t_n}\}$ exactly once for all $0 \leq j \leq n$ since $j+2^{t_n}$ always uses exactly $t_n+1$ bits for $j < 2^{t_n}$. In particular, $A_{n,j}$ is multilinear.\\

\begin{lemma}\label{lem:addressing-fns-partial-eval}
    Let $n,j \in \mathbb{N}$ with $0 \leq j \leq n$ and let $(b_0,...,b_{t_n}) \in \{0,1\}^{t_n+1}$, with $t_n$ as defined in $\autoref{def:addressing_fns}$. Then the following is true over any field, $\mathbb{F}$, of characteristic $p \neq 2$:
    \begin{enumerate}
        \item  $A_{n,j}(b_0,...,b_{t_n}) = \begin{cases}
            1 & \text{if $b_i =0$ for all $i \in B_{n,0}(j)$ and $b_i = 1$ for all $i \in B_{n,1}(j)$}, \\
            0 & \text{for all other choices of $(b_0,...,b_{t_g}) \in \{0,1\}^{t_g+1}$}
        \end{cases}$
        \item $A_{n,j}(\frac{1}{2},\frac{1}{2},...,\frac{1}{2},2^{t_n}) = 1$.
    \end{enumerate}
\end{lemma}
\begin{proof}
    \begin{enumerate}
        \item This is clear by construction.
        \item Since the bit corresponding to $y_{t_n}$ is always $1$ in the binary representation of $j+2^{t_n}$, we have that $y_{t_n}$ is a factor of $A_{n,j}$. Thus the evaluation becomes
        $$\left(\prod_{i \in B_{n,0}(j)}\left(1- \frac{1}{2}\right)\prod_{i \in B_{n,1}(j) \setminus \{t_n\}} \frac{1}{2} \right) \cdot 2^{t_n} = \frac{1}{2^{t_n}}\cdot 2^{t_n} = 1,$$
        since $|B_{n,0}(j)| + |B_{n,1}(j) \setminus \{t_n\}| = t_n$.
    \end{enumerate}
\end{proof}

The following lemma shows how these addressing gadgets are applied:\\

\begin{lemma}\label{lem:address-fns-application}
    Let $C(\mathbf{x})$ be a constant-free formula of size $s$ and depth $\Delta$  computing some polynomial $f(\mathbf{x)}$. We construct a new formula $C'(\mathbf{x},\mathbf{y})$ computing a new polynomial, $f'(\mathbf{x},\mathbf{y})$, as follows: 
    
    For any addition gate,
    $$g(\mathbf{x}) = \sum_{j=0}^{\operatorname{fanin}(g)-1} g_j(\mathbf{x})$$
    of $C(\mathbf{x})$, we replace $g$ by the subcircuit,
    $$g'(\mathbf{x},\mathbf{y}_g) = \sum_{j=0}^{\operatorname{fanin}(g)-1} g_j(\mathbf{x})\cdot A_{g,j}(\mathbf{y}_{g}),$$
    where we abuse notation and write $A_{g,j} := A_{\operatorname{fanin}(g)-1,j}$ and $t_g:=t_{\operatorname{fainin}(g)-1}$ and where $\mathbf{y}_{g} = \{y_{g,0},...,y_{g,t_g}\}$ is a fresh set of variables for each addition gate, $g$. Then 
    $$\mathbf{y} = \bigcup_{\text{$g$ addition gate in $C(\mathbf{x})$}} \mathbf{y}_{g},$$
    with $|\mathbf{y}| = \mathcal{O}(s \log s)$.
    We leave multiplication gates unchanged.

    Then the following are true over any field, $\mathbb{F}$, of characteristic $p \neq 2$:
    \begin{enumerate}
        \item $g'(\mathbf{x}, \mathbf{b}) = g_j(\mathbf{x})$ if $\mathbf{b}$ is the binary representation of $j+2^{t_g}$ as a vector.
        \item $f'(\mathbf{a},\mathbf{b}) \in \{-1,0,1\}$ for any choice of $(\mathbf{a},\mathbf{b}) \in \{0,1\}^{|\mathbf{x}|+|\mathbf{y}|}$. 
        \item There exists $\mathbf{b} \in \mathbb{F}^{|\mathbf{y}|}$ such that $f'(\mathbf{x}, \mathbf{b}) = f(\mathbf{x)}$.
        \item $C'(\mathbf{x},\mathbf{y})$ is of size $\mathcal{O}(s \log s)$ and depth at most $2\Delta+2$.
    \end{enumerate}
\end{lemma}
\begin{proof}
    \begin{enumerate}
        \item This follows directly from \autoref{lem:addressing-fns-partial-eval}.
        \item This follows directly from part 1 and induction on the circuit layers. The base case follows from the assumption that $C$ is constant-free. 
        \item For each addition gate, $g'$, we let 
        $$\mathbf{b}_{g} = \left(\frac{1}{2}, \frac{1}{2},\ldots,\frac{1}{2},2^{t_g}\right).$$
        Then by \autoref{lem:addressing-fns-partial-eval} every addition gate, $g'$, under this evaluation becomes
        $$ g'(\mathbf{x},\mathbf{b}_{g}) = \sum_{j=0}^{\operatorname{fanin}(g)-1} g_j(\mathbf{x})\cdot A_{g,j}(\mathbf{b}_{g}) = \sum_{j=0}^{\operatorname{fanin}(g)-1} g_j(\mathbf{x}) = g(\mathbf{x}).$$
        \item Let $n_g$ denote the fanin of gate $g$ (addition or multiplication). Since $A_{g,j}$ can be computed by formula of size $\mathcal{O}(\log n_g)$ and as $ n_g \leq s$  for any g, we get that $C'$ has a formula size at most $\mathcal{O}(s \log s)$.
        
        For each addition gate, $g$, we need to add a multiplication layer to multiply all the $g_j \cdot A_{g,j}$. Since $A_{g,j}$ has depth 2, we get a depth of at most $2\Delta + 2$.
    \end{enumerate}
\end{proof}

Only part 3 of \autoref{lem:address-fns-application} requires $\mathbb{F}$ to be of characteristic $p \neq 2$. The rest of the statement holds true over \emph{any} field.

\begin{remark}\label{rmk:01-out-only}
    If the input gates of $C(\mathbf{x})$ carry labels from $\mathbf{x} \cup \{0,1\}$, then $f'(\mathbf{a}, \mathbf{b}) \in \{0,1\}$ for any choice of $(\mathbf{a}, \mathbf{b}) \in \{0,1\}^{|\mathbf{x}| + |\mathbf{y}|}$. In particular, this applies to the formula of $Q_\Delta$ from \autoref{thm:lst-hierarchy}.
\end{remark}

\subsection{Proof}
\begin{proof}[Proof of \Cref{thm:main-constant-depth}]
Let $\mathbb{F}$ be a field of characteristic zero and fix any depth $\Delta \in \mathbb{N}$. Let $Q_{\Delta} \in \F[x_{1},\ldots,x_{n}]$ be the polynomial from \cite[Theorem 2.1]{LST} satisfying the following conditions:
\begin{enumerate}
    \item There is a constant-free algebraic formula computing $Q_{\Delta}$ in depth $\Delta$ and size $s$. Denote this  by $C(x_{1},\ldots,x_{n})$.
    \item For every algebraic formula computing $Q_{\Delta}$ with depth $\Delta/2-1$ requires size $s^{\omega(1)}$.
\end{enumerate}

\paragraph*{Unsatisfiable instance.} Let $C(\x)$ be the formula of depth $\Delta$ computing $Q_\Delta(\x)$. Let $C'(\mathbf{x},\mathbf{y})$ denote the formula we get after applying the process mentioned in \autoref{lem:address-fns-application} to $C(\x)$, 
and let $f_{\Delta}'(\mathbf{x}, \mathbf{y})$ denote the polynomial computed by $C'(\mathbf{x},\mathbf{y})$. By \autoref{lem:address-fns-application}, $C'$ is of depth $2 \Delta+2$ and size $\mathcal{O}(s\log s)$. 
Define the polynomial $f_{\Delta}$ as
\begin{align*}
    f_{\Delta}(\mathbf{x},\mathbf{y}) \; := \; f'_{\Delta}(\mathbf{x},\mathbf{y}) - 2.
\end{align*}

    Let $(\mathbf{a}, \mathbf{b}) \in \{0,1\}^{|\mathbf{x}|+|\mathbf{y}|}$ be any Boolean assignment to the variables in $f'_\Delta$. By \autoref{lem:address-fns-application} (and \autoref{rmk:01-out-only}), we get that $f'_\Delta(\mathbf{a},\mathbf{b}) \in \{0,1\}$, so $f_\Delta(\mathbf{a}, \mathbf{b}) \in \{-2,-1\}$. Hence, $f_\Delta =0 $ is not satisfiable over the Boolean cube. 
    Let $C_\Delta$ denote the formula for $f_\Delta$ of depth $2\Delta + 2$ and size $s_\Delta = O(s \log s)$. 

Now we show that $f_{\Delta}$ satisfies the two properties stated in \Cref{thm:main-constant-depth}. That is, we prove upper and lower bounds on the complexity of the refutation of $f_\Delta$. 

\subsubsection{Upper bound on the refutation of $f_\Delta$}
This section is dedicated to the proof of the following lemma.

\begin{lemma}[Upper Bound]\label{lemma:constant-depth-IPS-upper-bound}
   Let $f_{\Delta} \in \F[\mathbf{x},\mathbf{y}]$ be as defined above.  There exists a constant-depth $\IPS$ refutation of depth $\Delta' = 4\Delta + 6$ and size $s' \leq 100 s_\Delta^5$.
\end{lemma}

In order to prove the lemma, we use induction on the structure of  $C_\Delta$. We prove the following inductive lemma, which implies~\Cref{lemma:constant-depth-IPS-upper-bound}. 

\begin{lemma}
    \label{lem:const-depth-induction}
    Let $g$ be any gate in $C_\Delta$. Then,
    \[g^2-g = \sum_{i = 1}^n E_{g,i} \cdot (x_i^2 - x_i) + \sum_{j=1}^{m} F_{g,j}\cdot  (y_{j}^2-y_{j}),\]
    where $m$ is the number of $\y$ variables and 
    \begin{itemize}
        \item $\mathrm{size}(E_{g,i})$, $\mathrm{size}(F_{g,i})$ is at most $100 \cdot (\mathrm{size}(g))^4$, 
        \item and $\mathrm{depth}(E_{g,i})$, $\mathrm{depth}(F_{g,i})$ is at most $2 \cdot \mathrm{depth}(g)$.
    \end{itemize}
\end{lemma}

We assume \Cref{lem:const-depth-induction} and prove \Cref{lemma:constant-depth-IPS-upper-bound}. We know that $f_{\Delta}'$ can be computed by a formula of depth $2\Delta+2$ and size $s_\Delta$. Using \Cref{lem:const-depth-induction}, we know that there exists polynomials $E_{j}$ and $F_{j}$ such that
\begin{align*}
    (f_{\Delta}')^{2} - f_{\Delta}' \; = \; \sum_{j=1}^{n} E_{j}(\mathbf{x},\mathbf{y}) \cdot (x_{j}^{2}-x_{j}) \; + \; \sum_{j=1}^{m} F_{j}(\mathbf{x},\mathbf{y}) \cdot (y_{j}^{2}-y_{j}),
\end{align*}
where
\begin{itemize}
    \item For every $j$, the polynomial $E_{j}$ can be computed by a formula of depth $2(2\Delta+2)$ and size $100 \cdot s_\Delta^4$.
    \item For every $j$, the polynomial $F_{j}$ can be computed by a formula of depth $2(2\Delta+2)$ and size $100 \cdot s_\Delta^4$.
\end{itemize}
As $f_{\Delta} = f_{\Delta}' - 2$, we get,
\begin{align*}
    \dfrac{-1}{2} \paren{ (f_{\Delta}'(\mathbf{x},\mathbf{y}) + 1) \cdot f_{\Delta}(\mathbf{x},\mathbf{y}) \; + \; \sum_{j=1}^{n} E_{j}(\mathbf{x},\mathbf{y}) \cdot (x_{j}^{2}-x_{j}) \; + \; \sum_{j=1}^{m} F_{j}(\mathbf{x},\mathbf{y}) \cdot (y_{j}^{2}-y_{j}) } \; = \; 1.
\end{align*}
Thus we have an IPS refutation for $f_{\Delta}$ of depth $2(2\Delta +2) + 2$ and size at most $s_\Delta + O(1) + (n \cdot (100 \cdot s_\Delta^4 + O(1)) + (m \cdot (100 \cdot s_\Delta^4 + O(1))) $. As both $m$ and $n$ are bounded by $s_\Delta$, we get that this quantity is bounded by $100 s_\Delta^5$. This completes the proof\footnote{Note that we get a bound on the size of the Nullstellensatz refutation, thus a bound on $\IPSLIN$ refutation.} of~\Cref{lemma:constant-depth-IPS-upper-bound}.  
In what follows, we prove \Cref{lem:const-depth-induction}.

\begin{proof}[Proof of~\Cref{lem:const-depth-induction}]
    We will prove the statement by induction on the structure of $C_\Delta$. The base case is trivial. 
    For the induction step, we have two cases: either $g$ is a $\times$ gate or a $+$ gate. 
    
    \underline{case 1:} $g = \prod_{\ell = 0}^{r-1} g_\ell $.\\
    In this case,
    \begin{align*}
        g^2 -g  = & (g_0 \cdot g_1 \cdot \ldots \cdot g_{r-1})^2 -  (g_0 \cdot g_1 \cdot \ldots \cdot g_{r-1}) \\   
         = & (g_0 \cdot g_1 \cdot \ldots \cdot g_{r-1})^2 - g_0(g_1 \cdot \ldots \cdot g_{r-1})^2 + g_0(g_1 \cdot \ldots \cdot g_{r-1})^2 - (g_0 \cdot g_1 \cdot \ldots \cdot g_{r-1}) \\
         = & (g_0^2 - g_0) (g_1 \cdot \ldots \cdot g_{r-1})^2 +  \left((g_1 \cdot \ldots \cdot g_{r-1})^2- (g_1 \cdot \ldots \cdot g_{r-1})\right)
    \end{align*}

    Using the same idea as above, i.e., a telescoping  summation, we get
    \begin{equation}
        \label{eq:mult-induction}
        g^2- g = \sum_{\ell =0 }^{r-1} ~ \prod_{t < \ell} g_t \cdot 
    \prod_{t > \ell} g_t^2 \cdot (g_\ell^2-g_\ell), 
    \end{equation}
    where $t$ takes values between $0$ and $r-1$ and  $\prod_{t < \ell} g_t = 1$ if $\ell = 0$ and $\prod_{t > \ell} g_t^2 = 1$ if $\ell = r-1$. 
    Let $H_\ell = \prod_{t < \ell} g_t \cdot 
    \prod_{t > \ell} g_t^2$. Then note that $H_\ell$ has size at most $2 \cdot \mathrm{size}(g)$. 
    By induction hypothesis for $g_\ell^2-g_\ell$, we get that 
    \begin{align*}
        g^2-g = & \sum_{i=1}^n \underbrace{\sum_{\ell =0}^{r-1} H_\ell \cdot E_{g_\ell, i}}_{E_{g_i}} \cdot  (x_i^2 -x_i)
         + \sum_{j=1}^m \underbrace{\sum_{\ell =0}^{r-1} H_\ell \cdot F_{g_\ell, i}}_{F_{g,j}} \cdot  (y_j^2 -y_j) \\
    \end{align*}

    The above expression now allows us to bound the size and depth of $E_{g,i}$ for every $i \in [n]$ and the size of $F_{g,j}$ for every $j \in [m]$ as follows. 
    
    \textbf{Size bound.} Before we start the analysis, we recall that we measure the size of a formula by the number of wires in the formula. We also note some bounds on our parameters. We will assume that $s_g> 1$. Let $s_g$ be the short-hand for $\mathrm{size}(g)$ and for a fixed $g$ let $s_{g, \ell}$ denote $\mathrm{size}(g_\ell)$ for $0 \leq \ell \leq r-1$. 
    \begin{equation}
        \label{eq:size-bounds}
        \sum_\ell s_{g,\ell} \leq (s_g-r) \text{ and } \sum_\ell s_{g,\ell}^4 \leq (s_g-r)^4
    \end{equation}
    Moreover, we have for any parameter $s > 1$, $(s-1)^4 \leq s^4 - s^3/2$. 
    
    Now we will bound the size of $E_{g,i}$. In order to so that, we have already seen that size of $H_\ell$ is at most $s_g$. We can also bound the size of $E_{g_\ell,i}$ inductively. Thus,

    $\mathrm{size ~of~ } E_{g,i} \leq \left(\sum_{\ell} 2 s_g + 100 s_{g,\ell}^4\right) + 3r \leq 3 \cdot s_g \cdot r + 3r + 100 \cdot (s_g-1)^4$. 

    A similar bound can be proved on the size of $F_{g,j}$.

    The first bound comes from applying the bounds for $H_\ell$, $E_{g_\ell, i}$ and counting the wires feeding into the outer summation. The second bound comes from using~\Cref{eq:size-bounds} and the fact that $s_g-r \leq s_g-1$. 

    $\mathrm{size ~of~ } E_{g,i} \leq  6\cdot s_g \cdot r + 100 \cdot s_g^4 - 100 s_g^3/2 \leq 100 \cdot s^4$.

    \textbf{Depth bound.} $\mathrm{depth}(E_{g,i}) \leq 2 \cdot (\mathrm{depth}(E_{g_\ell,i})) + 2 \leq 2 \cdot (\mathrm{depth}(g) - 1) + 2 \leq 2 \cdot \mathrm{depth}(g)$. The depth of $F_{g,i}$ can be bounded similarly. 

    \underline{case 2:} $g = \sum_{\ell =0}^{r-1} g_{\ell} \cdot A_{g,\ell}$. \\
    In order to prove this case, we will make use of a couple of simple lemmas. 

    \begin{lemma}
        \label{lem:square-address}
        Let $g = \sum_{\ell =0}^{r-1} g_{\ell} \cdot A_{g,\ell}$, then for any $\ell$, $(A_{g,\ell})^2 - (A_{g,\ell}) = \sum_{j = 0}^{t_g} C_{\ell,j} \cdot (y_{g,j}^2 - y_{g,j})$, where size$(C_{\ell,j}) \leq 6 \cdot r$.
    \end{lemma}
    \begin{proof}
        Notice that for any addressing gadget attached to a gate $g$, it only uses variables from $\y_g$. For concreteness let the addressing gadget be given by $A_{g,\ell}(\y_g) = \prod_{t \in Y} y_{g,t} \times \prod_{t \in Y'} (1-y_{g,t}),$ for some partition of the indices of $\y_g$ into $Y$ and $Y'$. 

        Then, we get \[(A_{g,\ell})^2 - (A_{g,\ell}) = \left(\prod_{t \in Y} y_{g,t} \times \prod_{t \in Y'} (1-y_{g,t})\right)^2  - \left(\prod_{t \in Y} y_{g,t} \times \prod_{t \in Y'} (1-y_{g,t})\right)\] 

        Again using the idea of telescoping summations with respect to the $\mathbf{y}_{g}$ variables, we can show that 

        \[(A_{g,\ell})^2 - (A_{g,\ell}) = \sum_{j = 0}^{t_g} C_{\ell,j} \cdot (y_{g,j}^2 - y_{g,j})\]
        where, $C_{\ell, j}$ consists of monomials in $\y_g$ and it is a $\Pi\Sigma$ circuit in $\y_g$ variables. The input the product gates could be one of the following: either a variable appears as itself, or its square appears, or as $(1-y)$ or as $(1-y)^2$. The size of each linear factor can be bounded by $3\cdot t_g$ and hence, the overall size can be bounded by $6 \cdot t_g$. This is upper bounded by $6 \cdot r$. 
    \end{proof}
    
    \begin{lemma}
        \label{lem:prod-address}
        Let $g = \sum_{\ell =0}^{r-1} g_{\ell} \cdot A_{g,\ell}$, then for any $\ell \neq \ell'$, there exists a $j \in \{0,\ldots, t_g\}$ such that $A_{g, \ell} \times A_{g, \ell'} = C_{\ell, \ell', j} \cdot (y_j^2-y_j)$, where $\mathrm{size}(C_{\ell, \ell', j}) \leq 6 \cdot r$.
    \end{lemma}
    \begin{proof}
        Here, it is easy to observe that for $\ell \neq \ell'$, there must exist a variable $y_j$ such that either $y_j$ divides $A_{g,\ell}$ and $(1-y_j)$ divides $A_{g, \ell'}$ or vice-versa. Thus, $y_j \cdot (1-y_j)$ divides $A_{g,\ell} \times A_{g, \ell'}$. Thus, we get $A_{g, \ell} \times A_{g, \ell'} = C_{\ell, \ell', j} \cdot (y_j^2-y_j)$, where $C_{\ell, \ell',j}$ is simply the circuit consisting of a polynomial in $\y_g$ variables. 
        As in~\Cref{lem:square-address}, here again we get $\mathrm{size}(C_{\ell,\ell',j}) \leq 6 \cdot r$.          
    \end{proof}

    We will now resume the proof of case 2, i.e., the case when $g$ is a sum gate. We will again analyze $g^2-g$. 
    \begin{align*}
        g^2-g = & \left(\sum_{\ell =0}^{r-1} g_\ell \cdot A_{g, \ell}\right)^2 - \sum_{\ell =0}^{r-1} g_\ell \cdot A_{g, \ell} \\
        = & \sum_{\ell =0}^{r-1} \left(g_\ell \cdot A_{g, \ell}\right)^2 + \sum_{\ell \neq \ell'} g_\ell \cdot A_{g, \ell} \cdot g_{\ell'} \cdot A_{g, \ell'} - \sum_{\ell =0}^{r-1} g_\ell \cdot A_{g, \ell} \\
        = & \sum_{\ell =0}^{r-1} \left(g_\ell \cdot A_{g, \ell}\right)^2  - \sum_{\ell =0}^{r-1} g_\ell \cdot \left(A_{g, \ell}\right)^2 + \sum_{\ell =0}^{r-1} g_\ell \cdot \left(A_{g, \ell}\right)^2 - \sum_{\ell =0}^{r-1} g_\ell \cdot A_{g, \ell} + \sum_{\ell \neq \ell'} g_\ell \cdot g_{\ell'} \cdot A_{g, \ell} \cdot A_{g, \ell'},
    \end{align*}
    where we added and subtracted the same quantity (second and third summation). After rearranging, we get
    \begin{align*}
    g^2 -g = &  \sum_{\ell =0}^{r-1} \left(g_\ell^2-g_\ell \right)\cdot \left(A_{g, \ell}\right)^2 + \sum_{\ell =0}^{r-1} g_\ell \cdot \left(A_{g, \ell}^2 - A_{g, \ell} \right)+ \sum_{\ell \neq \ell'} g_\ell \cdot g_{\ell'} \cdot A_{g, \ell} \cdot A_{g, \ell'},
    \end{align*}
    We now apply induction on the first term and apply~\Cref{lem:square-address} and~\Cref{lem:prod-address} on the second and third terms, respectively. 
    \begin{align*}
        g^2-g 
        &= \sum_{i=1}^n \underbrace{\sum_{\ell=0}^{r-1} (A_{g, \ell})^2 \cdot  E_{g_\ell, i}}_{(E_{g,i})} \cdot (x_i^2-x_i) 
        + \sum_{j=1}^m \underbrace{\sum_{\ell=0}^{r-1} (A_{g, \ell})^2 \cdot F_{g_\ell, j}}_{(I)} \cdot (y_j^2 - y_j) \\
        &+ \sum_{j=0}^{t_g}  \underbrace{\sum_{\ell=0}^{r-1} C_{\ell,j} \cdot  g_\ell}_{(II)} \cdot (y_j^2-y_j) 
        + \underbrace{\sum_{\ell \neq \ell'} g_\ell \cdot g_{\ell'} \cdot C_{\ell, \ell', j}}_{(III)} \cdot (y_j^2-y_j) 
    \end{align*}

    Using this expression, we can now derive size and depth bounds on the refutation size. 

    \textbf{Size bound.} 

    We will now bound each term in the expression above. 

    \textbf{Bounding the size of $E_{g,i}$. }
    The size of $A_{g,\ell}$ can be bounded by $2r$. The size of each $E_{g_\ell, i}$ can be bounded inductively. Finally, the numbers of wires feeding into the outer summation can be bounded by $3r$. Thus we have, the following bound 
    \begin{align*}
        \mathrm{size}(E_{g,i}) \leq &
        \sum_{\ell =0}^{r-1} (2r + 100 s_{g,\ell}^4) + 3r \leq 2r^2 + 100 \sum_{\ell} s_{g,\ell}^4 + 3r  \leq 5r^2 + 100 \sum_{\ell} s_{g,\ell}^4 \\
        \leq & ~5r^2 + 100 (s-r)^4~~~~~ \text{Using~\Cref{eq:size-bounds}} \\
        \leq & 5r^2 + 100s^4 - 100 s^3/2 \leq 100 s^4. ~~~~~~\text{Using the bound on $s-r$ and $(s-1)^4$}
    \end{align*}
    
     \textbf{Bounding the size of $F_{g,i}$.} The bound on the size of $F_{g,i}$ can be obtained by analyzing terms (I), (II), and (III) above. Note that, the bound on term (I) is identical to the bound on $E_{g,i}$. So, we will have 
     
     \[\mathrm{ size~ of ~(I)} \leq 5r^2 + 100 \sum_{\ell} s_{g,\ell}^4 \quad \quad \quad (a). \]
     To bound the size of (II), we will use Lemma~\Cref{lem:square-address}. We get 
     \[\mathrm{size ~of ~(II)} \leq \left(\sum_{\ell} 6\cdot r + s_{g,\ell}\right) + 3 \cdot r \leq 9r^2 + \sum_{\ell} s_{g,\ell} \quad \quad \quad (b) \]
     Finally, we bound the size of (III). 
     \begin{align*}
         \mathrm{size ~of ~(III)} \leq & \left(\sum_{\ell \neq \ell'} 6\cdot r + s_{g,\ell} + s_{g,\ell'}\right) + 4 \cdot r^2 \\
          = & \sum_{\ell \neq \ell'} 6 \cdot r + \sum_{\ell \neq \ell'} s_{g,\ell} + \sum_{\ell \neq \ell'} s_{g,\ell'} + 4\cdot r^2 \\
          \leq & ~6r^3 + 2r \sum_{\ell \neq \ell'} s_{g,\ell} + 4r^2 \quad \quad \quad (c)
     \end{align*}
    Putting (a), (b), and (c) together and by combining terms we get that 

    $\mathrm{size}(F_{g,i}) \leq 25r^3 + 3rs + 100 \sum_{\ell} s_\ell^4 \leq 25r^3 + 3rs + 100 (s-1)^4 \leq  25r^3 + 3rs + 100 s^4 - 100s^3/2 \leq 100 s^4$

    \textbf{Depth bound.} The depth is bound is similar to the one we had in case 1 above. $\mathrm{depth}(E_{g,i}) \leq 2 \cdot (\mathrm{depth}(E_{g_\ell,i})) + 2 \leq 2 \cdot (\mathrm{depth}(g) - 1) + 2 \leq 2 \cdot \mathrm{depth}(g)$. The depth of $F_{g,i}$ can be bounded similarly. 
\end{proof}

\subsubsection{Lower bound on the refutation of $f_\Delta$}

In this section we focus on the lower bound of the size of the constant-depth $\IPS$ refutation of $f_\Delta$. Recall that $Q_\Delta(\mathbf{x})$ denotes the polynomial from \cite[Theorem 2.1]{LST} and let $s$ denote the size of the depth-$\Delta$ circuit computing $Q_\Delta$.
\begin{lemma}[Lower Bound]\label{lemma:constant-depth-IPS-lower-bound}
    Let $f_{\Delta}$ be as defined above. Every constant-depth $\IPS$ refutation of depth $\Delta''\leq \Delta/2-11$ requires size $s'' =s^{\omega(1)}$.
\end{lemma}
\begin{proof}
    Let $C''((\mathbf{x},\mathbf{y}),u,\vbar, \z)$ be an $\IPS$ refutation of $f_\Delta(\mathbf{x},\mathbf{y})$. Let $s''$ and $\Delta''$ denote the circuit size and the depth of $C''$, respectively. Then we have the following facts:
    \begin{enumerate}
        \item By definition $C''((\mathbf{x},\mathbf{y}), {0},\mathbf{0}, \mathbf{0})=0$ and $C''((\mathbf{x},\mathbf{y}), f_\Delta,\x^2-\x, \y^2-\y)=1$ so using \cite[Lemma~6.1]{FSTW21} we have that 
        $$1 - C''((\mathbf{x},\mathbf{y}),0,\mathbf{x}^2 - \mathbf{x},\mathbf{y}^2 - \mathbf{y}) = f_\Delta \cdot h, $$
        for some polynomial $h(\mathbf{x},\mathbf{y})$.
        \item By \cite[Theorem~1.1]{BKRRSS}, from $1-C''((\mathbf{x},\mathbf{y}),0,\mathbf{x}^2 - \mathbf{x},\mathbf{y}^2 - \mathbf{y})$ we can extract an algebraic formula for $f_\Delta(\mathbf{x},\mathbf{y})$ whose size is $\operatorname{poly}(s'')$ and whose depth is at most $\Delta''+10$. 
        \item By construction of $f_\Delta(\mathbf{x},\mathbf{y})$ and by \autoref{lem:address-fns-application} there exists some $\mathbf{b} \in \F^{|\mathbf{y|}}$ such that 
        $$f_\Delta(\mathbf{x},\mathbf{b}) = Q_\Delta(\mathbf{x}),$$
        so any formula computing $f_\Delta$ at depth $\Delta''+10$ also computes $Q_\Delta$ at depth $\Delta''+10$.
        \item \cite[Theorem 2.1]{LST} (also stated in \autoref{thm:lst-hierarchy}) states that every algebraic formula computing $Q_\Delta$ at depth $\Delta/2-1$ requires size $s^{\omega(1)}$.
    \end{enumerate}
    Putting all this together, we get that if $\Delta''+10 \leq \Delta/2-1$ 
    then $C''((\mathbf{x},\mathbf{y}),u, \vbar, \z)$ requires size $s^{\omega(1)}$.
\end{proof}

\begin{remark}
    We use the recent factorization result of~\cite{BKRRSS} to obtain our lower bound. However, we would like to also note that this is not necessary in our case. Due to a results of~\cite{CKS-19}, it is known that small-degree factors of any polynomial computed by constant-depth circuits/formulas of polynomial size can also be computed by constant-depth circuits/formulas of polynomial size. The hard polynomial from~\cite{LST} as well as our addressing gadgets have small degree, our hard instance has small degree (i.e. logarithmic in the number of variables). Thus, all its factors are also of small degree. That is,~\cite{CKS-19} is applicable in our case. 

    We present the proof using~\cite{BKRRSS} as it will help adapt our proof strategy to other scenarios more directly if we obtain strong algebraic complexity lower bounds in the future. 
\end{remark}

We now use~\Cref{lemma:constant-depth-IPS-upper-bound}
 and~\Cref{lemma:constant-depth-IPS-lower-bound} to finish the proof of~\Cref{thm:main-constant-depth}. Note that the depth of the $\IPS$ refutation for $f_\Delta$ is $4\Delta+6$ and size is $\poly(s)$, whereas any circuit of depth less than $\Delta/2-11$ requires superpolynomial size. Thus, for $\Gamma = 4\Delta+6$, we get that $f_\Gamma$ has an $\IPS$ refutation of depth $\Gamma$ and any $\IPS$ refutation of depth less than $\Gamma/10$ requires superpolynomial size. This completes the proof of~\Cref{thm:main-constant-depth}.  
\end{proof}

\section{Multilinear separation theorem}
\label{sec:mult}
We start by proving a lemma that will be useful in the rest of the section. 

\begin{lemma}
    \label{lem:mult-main}
    Let $N \in \N$ and let $\x = \{x_1, \ldots, x_N\}$.  Let $f(\mathbf{x})$ be a multilinear polynomial such that $f(\mathbf{a}) \in \{0,1\}$ for any $a \in \{0,1\}^N$, then 

    \begin{enumerate}
        \item The following identity holds. \[(2-f(\mathbf{x})) \times \frac{1+f(\mathbf{x})}{2} \equiv 1 \mod{\mathbf{x}^2 - \x} \]
        \item $2-f(\mathbf{x}) =0 $ is unsatisfiable over the Boolean cube. 
        \item The unique multilinear function $g(\mathbf{x})$ obeying $g(\mathbf{x}) (2-f(\mathbf{x})) \equiv 1 \mod{\mathbf{x}^2 - \mathbf{x}}$ is ${\left(1+f(\mathsf{x})\right)}/{2}$. Thus, it has the same multilinear circuit size and depth as $2-f(\mbox{x})$.         
    \end{enumerate}
\end{lemma}

\begin{proof}
    The first part of the lemma is a simple check. If $f(\mathbf{a}) = 1$ the left hand side evaluates to 1. Similarly, when $f(\mathbf{a}) = 0$ it again evaluates to 1.  The second part follows because we have assumed that $f(\x)$ only takes Boolean values. Finally, the third part follows immediately from the first part. 
\end{proof}

\subsection{Multilinear-$\NCo$ vs. multilinear-$\NCt$-IPS}
\label{sec:mnc1-mnc2}
In this section we prove~\Cref{thm:main-mnc1-mnc2}. 

We will construct a polynomial such that it is computable by a multilinear $\NCt$ circuit and such that it only takes Boolean values over the Boolean hypercube. This along with \Cref{lem:mult-main} will give us the desired separation.

\paragraph{Notation.} Let $[n] = \{1,\ldots, n\}$ and let $\ubar = \{u_1, \ldots, u_{2n}\}$ and $\vbar = \{v_{i,j,k}\}_{i,j,k\in [2n]}$. For $i,j \in [n]$ let $[i,j]$ denote the interval $\{k \mid i \leq k \text{ and } k \leq j\}$. Let $\ell([i,j])$ denote the length of the interval, i.e., $j-i +1$. When $j<i$, then $[i,j] = \emptyset$.

We first recall the hard polynomial defined by~\cite{RY08}, which is a simplification of the polynomial defined by~\cite{raz-ml-formulavscircuit}, which showed the first separation between multilinear formulas and circuits. 

The polynomial is defined inductively as follows. For $i \in [n]$, let $f_{i,i}(\ubar, \vbar) = 1$. If $\ell([i,j])$ is an even number more than $0$, then 
\[f_{i,j}(\ubar, \vbar) = (1 + u_i u_j) \cdot f_{i+1, j-1}(\ubar, \vbar) ~+~ \sum_{r \in [i+1,j-1]} v_{i,r,j} \cdot f_{i,r}(\ubar, \vbar) \cdot f_{r+1,j}(\ubar, \vbar) \]

Finally, the hard polynomial is $ F(\ubar, \vbar) = f_{1,2n}(\ubar, \vbar)$. They prove the following theorem about the polynomial.

\begin{lemma}[\cite{raz-ml-formulavscircuit}, \cite{RY08}]
    \label{lem:mnc1-mnc2}
    Let $n \in \N$ and let $\ubar = \{u_1, \ldots, u_{2n}\}$ and $\vbar = \{v_{i,j,k}\}_{i,j,k\in [2n]}$ be two sets of variables. Let $F(\ubar, \vbar)$ be the polynomial defined above. Then the following holds.
    \begin{enumerate}
        \item $F(\ubar, \vbar)$ can be computed by a multilinear circuit of size $\poly(n)$ and depth $O(\log^2 n)$.
        \item Any multilinear formula computing $F(\ubar,\vbar)$ must have size $n^{\Omega(\log n)}$. 
    \end{enumerate}
    
\end{lemma}

\begin{remark}
    Note that $F(\ubar, \vbar)$ when evaluated over the Boolean hypercube can take large values. (In fact, $F(\mathbf{1}, \mathbf{1})$ grows exponentially with $n$.) For~\Cref{lem:mult-main} to be applicable, we need a polynomial that takes only Boolean values over the hypercube. 
\end{remark}

We construct such a polynomial by modifying $F(\ubar, \vbar)$. We design the polynomial using new set of gadget variables, which we call $\wbar$. These will serve a dual purpose, first, they will help us create addressing gadgets for the $+$ gates and they will assume the role of $\vbar$ variables in the definition of $F(\ubar, \vbar)$.

    Let $\ubar = \{u_1, \ldots, u_{2n}\}$. For an interval $[i,j]$, let $W^{[i,j]}$ denote the following set of variables. \[W^{[i,j]} = \{w^{[i,j]}_{\text{top}}\} \cup \{w^{[i,j]}_{\text{leaf}}\} \cup \tilde{W}^{[i,j]}, \]
    where $\tilde{W}^{[i,j]}$ consists of variables we will use for the addressing gadgets. Recall that for $n \in \mathbb{N}$, $t_n$ denotes the smallest integer such that $2^{t_n} > n$. Let $t_{i,j}$ be the shorthand for $t_{\ell([i+1,j-1])}$ and $n_{i,j}$ be a shorthand for $\ell([i+1,j-1])$.  Let $ \tilde{W}^{[i,j]}= \{w^{[i,j]}_{r} \mid 0 \leq r  \leq t_{i,j}\}$.

    Finally, we define $\wbar = \bigcup_{[i,j]} W^{[i,j]}$, where the union is over all intervals $[i,j]$, where $1 \leq i <j \leq n$ and $\ell([i,j])$ is even\footnote{We only use intervals of even length inductively.}. Here, the size of any set $W^{[i,j]}$ is $O(\log n)$ and hence we have $O(n^2 \log n)$-many $w$ variables. We will use $m(n)$ to denote the cardinality of $\wbar$ and $m$, when $n$ is clear from the context.

\begin{definition}
\label{def:mult-nc-hardpoly}
    The hard polynomial is defined inductively as follows:
    If $\ell([i,j]) = 0$ then $p_{[i,j]}(u,w) = 1$. If $\ell([i,j]) >0$ and even, then

    \[
    \begin{array}{cc}
    p_{i,j}(\ubar, \wbar)= & (1-w^{[i,j]}_{\text{top}}) \left((1-w^{[i,j]}_{\text{leaf}}) + w^{[i,j]}_{\text{leaf}} \cdot u_i \cdot u_j \right) \times p_{i+1,j-1} \\
     & + ~~w^{[i,j]}_{\text{top}} \times \left(\sum_{r \in [i+1, j-1]} g_r(\tilde{W}^{[i,j]}) \times p_{i,r} \cdot p_{r+1, j}\right),\\
    \end{array}
    \]
    $g_r$ is the addressing gadget, i.e. 
    \[g_r(\tilde{W}^{[i,j]}) = \prod_{t \in B_{n_{i,j},0}(r)} (1-w^{[i,j]}_{t}) \cdot \prod_{t \in B_{n_{i,j},1}(r)} w^{[i,j]}_{t}  \]
    Here, the sets $B_{n_{i,j},0}$ and $B_{n_{i,j},1}$ are defined as in~\Cref{def:addressing_fns}. 
    
    Finally, $P(\ubar, \wbar) = p_{1,2n}(\ubar, \wbar)$. 
\end{definition}

We will now prove that the polynomial $P(\ubar, \wbar)$ defined above retains the properties of the polynomial $F(\ubar, \vbar)$, that is, it is computed by multilinear circuits and it is hard for multilinear formulas. Additionally, we will show that the polynomial only takes Boolean values over the Boolean hypercube. Formally, we prove the following theorem. \\

\begin{theorem}
    \label{thm:mnc1-mnc2-hardpoly}
    Let $n \in \N$ and let $\ubar$ and $\wbar$ be as defined above. Also, let $P(\ubar, \wbar)$ be the polynomial from~\Cref{def:mult-nc-hardpoly}. Then, the following statements hold.
    \begin{enumerate}
        \item $P(\ubar, \wbar) \in \{0,1\}$ when evaluated over the Boolean hypercube. 
        \item $P(\ubar, \wbar)$ can be computed by a multilinear circuit of size $\poly(n)$ and depth $O(\log^2 n)$. 
        \item Any multilinear formula computing $P(\ubar, \wbar)$ must have size $n^{\Omega(\log n)}$.
    \end{enumerate}
\end{theorem}

Before we present the proof for the theorem, we will use it to prove our main theorem~\Cref{thm:main-mnc1-mnc2}, which we recall below.\\

\mainthmmultNC*

\begin{proof}[Proof of~\Cref{thm:main-mnc1-mnc2}]
    Let $N = n + m$, where $n$ is the cardinality of $\ubar$ and $m$ is the cardinality of $\wbar$ and let $\x = \ubar \cup \wbar$. We will define $Q(\x) = 2-P(\x)$, where $P$ is as in~\Cref{thm:mnc1-mnc2-hardpoly}. Clearly $Q(\x)$ is unsatisfiable over the Boolean hypercube. Using~\Cref{lem:mult-main}, we know that the refutation for $Q(\x)$ is $((P(\x))+1)/2$ modulo the Boolean axioms. 
    Thus, the functional refutation of $Q(\x)$ is computable by multilinear circuit of size $\poly(n)$ and depth $O(\log^2 n)$. 

    Moreover, from~\Cref{thm:mnc1-mnc2-hardpoly} we know that $P(\x)$ requires multilinear $\NCo$ circuit of size $n^{\Omega(\log n)}$. As the refutation is $(P(\x) +1)/2$, we also get that the multilinear $\NCo$-$\IPSLINp$ refutation\footnote{As we use functional method, we get a lower bound in $\IPSLINp$ and not just for $\IPSLIN$.} for it must have size $n^{\Omega(\log n)}$. As $N$ and $n$ are polynomially related, this also gives a $N^{\Omega(\log N)}$ lower bound on the proof size of multilinear $\NCo$-$\IPSLIN$ refutations. 
\end{proof}
    
\subsection{Proof of~\Cref{thm:mnc1-mnc2-hardpoly}}

\paragraph{Part 1 of~\Cref{thm:mnc1-mnc2-hardpoly}.} We prove this statement by using the inductive structure of $P(\ubar, \wbar)$. Specifically, we will show that for any interval $[i,j]$, the polynomial corresponding to it, i.e., $p_{i,j}(\ubar, \wbar) \in \{0,1\}$ when evaluated over the Boolean hypercube. We induct on the length of the interval. We only need to consider even length intervals.

\textbf{Base case.} Suppose $\ell([i,j]) =0$ then the statement trivially holds.

\textbf{Inductive step.}
Suppose $\ell([i,j]) >0$ \mnote{and even?}. The polynomial $p_{i,j}$ is as defined in~\Cref{def:mult-nc-hardpoly}. 

Suppose $w^{[i,j]}_{\text{top}} =0$, then \[p_{i,j} = \left((1-w^{[i,j]}_{\text{leaf}}) + w^{[i,j]}_{\text{leaf}} \cdot u_i \cdot u_j \right) \cdot p_{i+1, j-1}.\]
Notice that if $w^{[i,j]}_{\text{leaf}} =0$ then $p_{i,j} = p_{i+1, j-1}$, which by induction hypothesis is Boolean. If $w^{[i,j]}_{\text{leaf}} = 1$ then $p_{i,j} = u_i u_j p_{i+1, j-1}$, which is either $0$ or $1$ for Boolean values of $u_i$, $u_j$ and $p_{i+1, j-1}$.  

On the other hand, if $w^{[i,j]}_{\text{top}} =1$, then 
\[p_{i,j} = \left(\sum_{r \in [i+1, j-1]} g_r(\tilde{W}^{[i,j]}) \times p_{i,r} \cdot p_{r+1, j}\right).\]

Now, suppose the variables in the set $\tilde{W}^{[i,j]}$ are set to $0$s and $1$s such that the boolean assignment equals $r+2^{t_{i,j}}$, then we get $p_{i,j} = p_{i,r} \cdot p_{r+1,j}$. By inductive assumption $p_{i,r} \in \{0,1\} $  and $ p_{r+1,j} \in \{0,1\}$. That finishes the proof. 
\paragraph{Part 2 of~\Cref{thm:mnc1-mnc2-hardpoly}.} Notice that polynomial $F(\ubar, \vbar)$ defined by~\cite{RY08} is very similar to $P(\ubar, \wbar)$. Instead of $\vbar$ variables, we have small local changes using the addressing gadgets. The addressing gadgets themselves are constant-depth (unbounded fan-in) multilinear circuits. It is easy to see that by implementing the inductive definition of $p(\ubar,\vbar)$, we will obtain a polynomial size and polynomial depth multilinear circuit. By using a depth reduction result of~\cite{RY08}, we can obtain a polynomial size and $O(\log^2n)$-depth multilinear circuit for the polynomial.

\paragraph{Part 3 of~\Cref{thm:mnc1-mnc2-hardpoly}.}
Firstly, we will prove that for every partition of variables in $\ubar$ into two sets of equal size say, $\ubar = \y \cup \z$, the rank of the matrix $M_{\y,\z}(p_{1,2n})$ is equal to $2^n$. The bound will then imply  a lower bound for the IPS proof size. This step is quite standard, but we will present it for completeness. 

\begin{lemma}
    \label{lem:ranklb-mnc1}
    Let $n \in \N$  and $p_{1,2n}$ be as defined above. Then, for any partition of $\ubar$ into $\y \cup \z$ each of cardinality $n$,  there exists an assignment to variables in $\wbar$ to field constants, such that rank$\left(M_{\y,\z}(p_{1, 2n})\right) = 2^{n}$. 
\end{lemma}
\begin{proof}   
    We will prove this by induction on $n$. 
    
    \textbf{Base case.} Suppose $n=1$, then $p_{1,2} = \left((1-w^{[1,2]}_{\text{leaf}}) + w^{[1,2]}_{\text{leaf}} \cdot u_1 \cdot u_2 \right) $. By setting $w^{[1,2]}_{\text{leaf}} = 1/2$, we get $p_{1,2} = \frac{1}{2} \cdot (1+u_1\cdot u_2)$ and the statement trivially holds. 

    \textbf{Inductive step.} Let $n>1$. We consider two cases. Either $u_1$ and $u_{2n}$ are in the same part under the partition of $\ubar$ into $\y \cup \z$ or they are in different parts. 

    \underline{$u_1$ and $u_{2n}$ in different parts} We will set $w^{[1,2n]}_{\text{top}} =0$ and $w^{[1,2n]}_{\text{leaf}} = {1}/{2}$. Under this substitution, $p_{1,2n} = \frac{1}{2}\left( 1 + u_1 \cdot u_{2n} \right) \cdot p_{2, 2n-1}.$
    By induction hypothesis, $p_{2, 2n-1}$ is full rank, i.e. $2^{n-1}$, under every equi-sized partition of its variables. And the rank for $(1+u_1u_{2n})$ is $2$. Note also that $p_{2, 2n-1}$ does not use the variables $u_1$ and $u_{2n}$.  Hence, we are done in this case.  

    \underline{$u_1$ and $u_{2n}$ in same part} In this case, there is an $r \in [i+1, j-1]$ such that the intervals $[1,r]$ and $[r+1, 2n]$ evenly split $\y$ and $\z$ variables. We set $w^{[1,2n]}_{\text{top}} =1$ and the variables in the addressing gadget to the binary encoding of $r+2^{t_{i,j}}$. This gives $p_{1,2n} = p_{1,r} \cdot p_{r+1,2n}$. Using induction on $p_{1,r}, p_{r+1,2n}$ and observing that the two polynomials do not share any variables we get the desired bound on the rank of $p_{1,2n}$. 
\end{proof}

\medskip

\printbibliography[
heading=bibintoc,
title={References}
] 

@article{RY09,
title = "Lower Bounds and separations for constant depth multilinear circuits",
author = "Ran Raz and Amir Yehudayoff",
year = "2009",
month = jun,
doi = "10.1007/s00037-009-0270-8",
language = "English",
volume = "18",
pages = "171--207",
journal = "Computational Complexity",
issn = "1016-3328",
publisher = "Birkhauser Verlag Basel",
number = "2",
}

@article{FSTW21,
  author       = {Michael A. Forbes and
                  Amir Shpilka and
                  Iddo Tzameret and
                  Avi Wigderson},
  title        = {Proof Complexity Lower Bounds from Algebraic Circuit Complexity},
  journal      = {Theory Comput.},
  volume       = {17},
  pages        = {1--88},
  year         = {2021},
  url          = {https://theoryofcomputing.org/articles/v017a010/},
  bibsource    = {dblp computer science bibliography, https://dblp.org}
}

@INPROCEEDINGS{GHT,
  author={Govindasamy, Nashlen and Hakoniemi, Tuomas and Tzameret, Iddo},
  booktitle={2022 IEEE 63rd Annual Symposium on Foundations of Computer Science (FOCS)}, 
  title={Simple Hard Instances for Low-Depth Algebraic Proofs}, 
  year={2022},
  volume={},
  number={},
  pages={188-199},
  doi={10.1109/FOCS54457.2022.00025}}

@INPROCEEDINGS{LST,
  author={Limaye, Nutan and Srinivasan, Srikanth and Tavenas, Sébastien},
  booktitle={2021 IEEE 62nd Annual Symposium on Foundations of Computer Science (FOCS)}, 
  title={Superpolynomial Lower Bounds Against Low-Depth Algebraic Circuits}, 
  year={2021},
  volume={},
  number={},
  pages={804-814},
  doi={10.1109/FOCS52979.2021.00083}}

@inproceedings{HLT24,
author = {Hakoniemi, Tuomas and Limaye, Nutan and Tzameret, Iddo},
title = {Functional Lower Bounds in Algebraic Proofs: Symmetry, Lifting, and Barriers},
year = {2024},
isbn = {9798400703836},
publisher = {Association for Computing Machinery},
address = {New York, NY, USA},
url = {https://doi.org/10.1145/3618260.3649616},
doi = {10.1145/3618260.3649616},
booktitle = {Proceedings of the 56th Annual ACM Symposium on Theory of Computing},
pages = {1396–1404},
numpages = {9},
location = {Vancouver, BC, Canada},
series = {STOC 2024}
}

@article{GP14,
author = {Grochow, Joshua A. and Pitassi, Toniann},
title = {Circuit Complexity, Proof Complexity, and Polynomial Identity Testing: The Ideal Proof System},
year = {2018},
issue_date = {December 2018},
publisher = {Association for Computing Machinery},
address = {New York, NY, USA},
volume = {65},
number = {6},
issn = {0004-5411},
url = {https://doi.org/10.1145/3230742},
doi = {10.1145/3230742},
journal = {J. ACM},
month = nov,
articleno = {37},
numpages = {59}
}

@article{CR79,
title={The relative efficiency of propositional proof systems},
volume={44},
DOI={10.2307/2273702},
number={1},
journal={Journal of Symbolic Logic},
author={Cook, Stephen A. and Reckhow, Robert A.},
year={1979},
pages={36–50}
}

@inproceedings{AF22,
author = {Andrews, Robert and Forbes, Michael A.},
title = {Ideals, determinants, and straightening: proving and using lower bounds for polynomial ideals},
year = {2022},
isbn = {9781450392648},
publisher = {Association for Computing Machinery},
address = {New York, NY, USA},
url = {https://doi.org/10.1145/3519935.3520025},
doi = {10.1145/3519935.3520025},
booktitle = {Proceedings of the 54th Annual ACM SIGACT Symposium on Theory of Computing},
pages = {389–402},
numpages = {14},
location = {Rome, Italy},
series = {STOC 2022}
}

@article{CKS-19,
  author       = {Chi{-}Ning Chou and
                  Mrinal Kumar and
                  Noam Solomon},
  title        = {Closure Results for Polynomial Factorization},
  journal      = {Theory of Computing},
  volume       = {15},
  pages        = {1--34},
  year         = {2019},
  url          = {https://doi.org/10.4086/toc.2019.v015a013},
  doi          = {10.4086/TOC.2019.V015A013},
  bibsource    = {dblp computer science bibliography, https://dblp.org}
}

@inproceedings{raz-ml-formulavscircuit,
author = {Raz, Ran},
title = {Multilinear- $\mathrm{NC}_{1}$ $\neq$ Multilinear- $\mathrm{NC}_{2}$},
year = {2004},
isbn = {0769522289},
publisher = {IEEE Computer Society},
address = {USA},
url = {https://doi.org/10.1109/FOCS.2004.42},
doi = {10.1109/FOCS.2004.42},
booktitle = {Proceedings of the 45th Annual IEEE Symposium on Foundations of Computer Science},
pages = {344–351},
numpages = {8},
series = {FOCS '04}
}

@inproceedings{CELS-18-multilinear-hierarchy,
  author       = {Suryajith Chillara and
                  Christian Engels and
                  Nutan Limaye and
                  Srikanth Srinivasan},
  editor       = {Mikkel Thorup},
  title        = {A Near-Optimal Depth-Hierarchy Theorem for Small-Depth Multilinear
                  Circuits},
  booktitle    = {59th {IEEE} Annual Symposium on Foundations of Computer Science, {FOCS}
                  2018, Paris, France, October 7-9, 2018},
  pages        = {934--945},
  publisher    = {{IEEE} Computer Society},
  year         = {2018},
  url          = {https://doi.org/10.1109/FOCS.2018.00092},
  doi          = {10.1109/FOCS.2018.00092},
  bibsource    = {dblp computer science bibliography, https://dblp.org}
}

@article{BKRRSS,
  author       = {Somnath Bhattacharjee and
                  Mrinal Kumar and
                  Shanthanu S. Rai and
                  Varun Ramanathan and
                  Ramprasad Saptharishi and
                  Shubhangi Saraf},
  title        = {Closure under factorization from a result of {F}urstenberg},
  journal      = {CoRR},
  volume       = {abs/2506.23214},
  year         = {2025},
  doi          = {10.48550/ARXIV.2506.23214},
  eprinttype    = {arXiv},
  eprint       = {2506.23214},
}

@inproceedings{ST23,
  author    = {Rahul Santhanam and Iddo Tzameret},
  title     = {Iterated Lower Bound Formulas: A Diagonalization-Based Approach to Proof Complexity},
  booktitle = {Proceedings of the 53rd Annual {ACM} {SIGACT} Symposium on Theory of Computing (STOC 2021)},
  publisher = {ACM},
  year      = {2021},
  pages     = {234--247},
  doi       = {10.1145/3406325.3451010},
  note      = {ECCC TR21-138; journal version in \emph{SIAM J. Comput.}, 2025}
}

@techreport{LuST25,
  author       = {Jiaqi Lu and Rahul Santhanam and Iddo Tzameret},
  title        = {AC\^{}0[p]-Frege Cannot Efficiently Prove that Constant-Depth Algebraic Circuit Lower Bounds are Hard},
  institution  = {Electronic Colloquium on Computational Complexity (ECCC)},
  number       = {TR25-134},
  year         = {2025},
  month        = sep,
  url          = {https://eccc.weizmann.ac.il/report/2025/134/}
}

@inproceedings{BLRS25,
  author =	{Behera, Amik Raj and Limaye, Nutan and Ramanathan, Varun and Srinivasan, Srikanth},
  title =	{{New Bounds for the Ideal Proof System in Positive Characteristic}},
  booktitle =	{52nd International Colloquium on Automata, Languages, and Programming (ICALP 2025)},
  pages =	{22:1--22:20},
  series =	{Leibniz International Proceedings in Informatics (LIPIcs)},
  ISBN =	{978-3-95977-372-0},
  ISSN =	{1868-8969},
  year =	{2025},
  volume =	{334},
  editor =	{Censor-Hillel, Keren and Grandoni, Fabrizio and Ouaknine, Jo\"{e}l and Puppis, Gabriele},
  publisher =	{Schloss Dagstuhl -- Leibniz-Zentrum f{\"u}r Informatik},
  address =	{Dagstuhl, Germany},
  URL =		{https://drops.dagstuhl.de/entities/document/10.4230/LIPIcs.ICALP.2025.22},
  URN =		{urn:nbn:de:0030-drops-233992},
  doi =		{10.4230/LIPIcs.ICALP.2025.22}
}

@misc{EGLT25,
  author       = {Tal Elbaz and Nashlen Govindasamy and Jiaqi Lu and Iddo Tzameret},
  title        = {Lower Bounds against the Ideal Proof System in Finite Fields},
  year         = {2025},
  month        = jun,
  eprint       = {2506.17210},
  archivePrefix= {arXiv},
  primaryClass = {cs.CC},
  doi          = {10.48550/arXiv.2506.17210},
  url          = {https://arxiv.org/abs/2506.17210}
}

@misc{CGMS25,
  author       = {Prerona Chatterjee and Utsab Ghosal and Partha Mukhopadhyay and Amit Sinhababu},
  title        = {IPS Lower Bounds for Formulas and Sum of ROABPs},
  year         = {2025},
  month        = jul,
  eprint       = {2507.09515},
  archivePrefix= {arXiv},
  primaryClass = {cs.CC},
  doi          = {10.48550/arXiv.2507.09515},
  url          = {https://arxiv.org/abs/2507.09515}
}

@article{Ladner,
  author    = {Juris Hartmanis and Richard E. Stearns},
  title     = {On the Computational Complexity of Algorithms},
  journal   = {Transactions of the American Mathematical Society},
  volume    = {117},
  pages     = {285--306},
  year      = {1965},
  month     = may,
  doi       = {10.2307/1994208},
  ISSN      = {0002-9947}
}

@inproceedings{sp-hierarchy,
  author    = {Richard E. Stearns and Juris Hartmanis and Philip M. Lewis II},
  title     = {Hierarchies of Memory Limited Computations},
  booktitle = {Proceedings of the 6th Annual Symposium on Switching Circuit Theory and Logical Design (FOCS)},
  publisher = {IEEE},
  year      = {1965},
  pages     = {179--190},
  doi       = {10.1109/FOCS.1965.11}
}

@inproceedings{HLT,
    author = {Hakoniemi, Tuomas and Limaye, Nutan and Tzameret, Iddo},
    title = {Functional Lower Bounds in Algebraic Proofs: Symmetry, Lifting, and Barriers},
    year = {2024},
    isbn = {9798400703836},
    publisher = {Association for Computing Machinery},
    address = {New York, NY, USA},
    url = {https://doi.org/10.1145/3618260.3649616},
    doi = {10.1145/3618260.3649616},
    booktitle = {Proceedings of the 56th Annual ACM Symposium on Theory of Computing},
    pages = {1396–1404},
    numpages = {9},
    location = {Vancouver, BC, Canada},
    series = {STOC 2024}
}

@Article{RY08,
    author={Raz, Ran
    and Yehudayoff, Amir},
    title={Balancing Syntactically Multilinear Arithmetic Circuits},
    journal={computational complexity},
    year={2008},
    month={12},
    day={01},
    volume={17},
    number={4},
    pages={515-535},
    issn={1420-8954},
    doi={10.1007/s00037-008-0254-0},
    url={https://doi.org/10.1007/s00037-008-0254-0}
}

@InProceedings{Andrews2025,
  author =	{Andrews, Robert},
  title =	{{Algebraic Pseudorandomness in $\text{VNC}^0$}},
  booktitle =	{40th Computational Complexity Conference (CCC 2025)},
  pages =	{15:1--15:15},
  series =	{Leibniz International Proceedings in Informatics (LIPIcs)},
  ISBN =	{978-3-95977-379-9},
  ISSN =	{1868-8969},
  year =	{2025},
  volume =	{339},
  editor =	{Srinivasan, Srikanth},
  publisher =	{Schloss Dagstuhl -- Leibniz-Zentrum f{\"u}r Informatik},
  address =	{Dagstuhl, Germany},
  URL =		{https://drops.dagstuhl.de/entities/document/10.4230/LIPIcs.CCC.2025.15},
  URN =		{urn:nbn:de:0030-drops-237092},
  doi =		{10.4230/LIPIcs.CCC.2025.15}
}

@book{Jukna12,
  author    = {Stasys Jukna},
  title     = {Boolean Function Complexity: Advances and Frontiers},
  series    = {Texts in Theoretical Computer Science},
  publisher = {Springer},
  year      = {2012},
  isbn      = {978-3-642-24507-7},
  note      = {Chs. 1--2 survey Shannon’s lower bound and Lupanov’s matching upper bound yielding a nonuniform circuit size hierarchy}
}

@ARTICLE{shannon49,
  author={Shannon, Claude. E.},
  journal={The Bell System Technical Journal}, 
  title={The synthesis of two-terminal switching circuits}, 
  year={1949},
  volume={28},
  number={1},
  pages={59-98},
  doi={10.1002/j.1538-7305.1949.tb03624.x}
}

@inproceedings{sipser83,
author = {Sipser, Michael},
title = {Borel sets and circuit complexity},
year = {1983},
isbn = {0897910990},
publisher = {Association for Computing Machinery},
address = {New York, NY, USA},
url = {https://doi.org/10.1145/800061.808733},
doi = {10.1145/800061.808733},
booktitle = {Proceedings of the Fifteenth Annual ACM Symposium on Theory of Computing},
pages = {61–69},
numpages = {9},
series = {STOC '83}
}

@inproceedings{hasted86,
author = {Hastad, J},
title = {Almost optimal lower bounds for small depth circuits},
year = {1986},
isbn = {0897911938},
publisher = {Association for Computing Machinery},
address = {New York, NY, USA},
url = {https://doi.org/10.1145/12130.12132},
doi = {10.1145/12130.12132},
booktitle = {Proceedings of the Eighteenth Annual ACM Symposium on Theory of Computing},
pages = {6–20},
numpages = {15},
location = {Berkeley, California, USA},
series = {STOC '86}
}

@article{Burgisser,
    author = {B{\"u}rgisser, Peter},
    year = {1998},
    month = {08},
    pages = {},
    title = {Completeness and Reduction in Algebraic Complexity Theory},
    volume = {7},
    isbn = {9783642086045},
    doi = {10.1007/978-3-662-04179-6_1}
}

\addtocontents{toc}{\protect\setcounter{tocdepth}{1}}

\appendix
\section{The complexity of refuting lifted subset-sum}
\label{sec:appendix}

Let $F(\x,\z) = \sum_{i<j \in [n]} z_{i,j}x_i x_j -\beta$ be the lifted subset-sum instance, where $\beta \in \Theta(n^3)$. Clearly, it is an unsatisfiable instance. It was used in~\cite{FSTW21} to prove a lower bound on the the size of the multilinear formula $\IPSLINp$. 
Here, we further analyze the hardness of refuting this instance. We show that its refutation must have high complexity under a standard complexity assumption. Specifically, we prove the following.\\

\begin{lemma}
    \label{lem:subset-sum-hard}
    If $F(\x,\z)$ has a polynomial size multilinear $\IPSLINp$ refutation, then $\mathsf{VP = VNP}$.
\end{lemma}

This makes our hard instance in~\Cref{thm:main-mnc1-mnc2} interesting. On the one hand we obtain an equally strong lower bound as in~\cite{FSTW21}, and on the other hand we also obtain a reasonably good upper bound on the functional refutation of our instance.

In order to prove the lemma, we start with some notation and preliminaries. Let $V \subseteq [n]$ and let $K_V = \{(i,j) \mid i, j \in V, i <j\}$. Let $e = (i,j)$ denote a pair from the set $K_V$, then we use $z_e$ to denote $z_{i,j}$. 
\begin{lemma}[\cite{Burgisser}]
    \label{lem:clique}
    Let $C_\ell(\x,\z)$ be the Clique polynomial defined as follows. 
    \[C_\ell(\x,\z) = \sum_{V \subseteq [n], |V| = \ell} ~~~\prod_{e \in  K_V} z_e \prod_{i \in V} x_i\]
    $C_{n/2}(\x,\z)$ is $\mathsf{VNP}$ complete\footnote{The Clique polynomial from~\cite{Burgisser} slightly differs from the polynomial we have here. Namely, it is $C_\ell(\z) = \sum_{V \subseteq [n], |V| = \ell}\prod_{e \in  K_V} z_e$. However, by substituting $x_is =1$ in the above polynomial, we can obtain this polynomial.}. 
\end{lemma}

We are now ready to prove~\Cref{lem:subset-sum-hard}. 
\begin{proof}[Proof of~\Cref{lem:subset-sum-hard}]
We introduce some notation. We use $\nCtwo$ to denote the set $\{(i,j) \mid i, j \in [n], i < j\}$. 
Let the subset-sum instance (without the lift) be \[f(\z) = \sum_{ (i,j) \in \nCtwo} z_{i,j} - \beta = 0\] for $\beta>n^2$ 

. (This instance is the same as the subset-sum instance in \cite[Section 5]{FSTW21} up to relabeling.) In \cite[Proposition B.1]{FSTW21}, they gave an explicit description of its multilinear functional refutation, i.e. they exactly computed the multilinear polynomial $g(\z)$ such that 
\begin{align*}
    g(\z)  = \dfrac{1}{\left(\sum_{ (i,j)\in \nCtwo} z_{i,j} - \beta \right)}, && \text{ for every } \mathbf{z} \in \Boo^{n}.
\end{align*}

They showed that every functional refutation of $f(\mathbf{z})$ can be expressed as a linear combination of the elementary symmetric polynomials of degree $k$, for every $k \in [n]$. More precisely, they showed that

\begin{equation}
\label{eq:fstw-refutation}
    g(\z) = \sum_{k=0}^n \alpha_k ~\cdot \sum_{\substack{S \subseteq \nCtwo \\ |S|=k} } \; \; \prod_{(i,j) \in S} z_{i,j},
\end{equation}
where $\alpha_k$ is a non-zero constant that only depends on $k$ and $\beta$. 

Now, we will first change the input instance to $F(\x,\z) = \sum_{(i,j)} z_{i,j} x_i x_j$, where the sum is over the set $\nCtwo$. As $F(\x,\z)$ can be obtained from $f(\z)$ by a monomial substitution $z_{i,j} \mapsto z_{i,j} x_i x_j$, it is easy to see that the functional refutation of $F(\x,\z)$ can be obtained from the refutation of $f(\z)$ by monomial substitution.\anote{This part is a bit unclear to me... Also we need to mention $F(\mathbf{x},\mathbf{z})$ is unsatisfiable} This is because we only need to preserve the refutation over the Boolean cube. \mnote{maybe more formally: $\ml[f]  = f \mod \mathbf{x}^2-\mathbf{x}$} Such a monomial substitution can result in a non-multilinear polynomial. Let $\ml[\cdot ]$ denote the following map defined for monomials over a set of variables, say $\y$: $\ml \left[\prod_i y_i^{a_i}\right] = \prod_{i} y_i^{\text{min}\{a_i, 1\}}$. The map extends linearly and can be defined as a map from $\F[\y] \rightarrow \F[\y]$ for any polynomial ring $\F[\y]$. 

Let $G(\x,\z)$ denote the unique multilinear \anote{we should say multilinear refutation to say it is unique} refutation of $F(\x,\z)$. Then, using~\Cref{eq:fstw-refutation} we get
\[G(\x,\z) = \sum_{k=0}^n \alpha_k ~\cdot \sum_{S \subseteq \nCtwo, |S|=k} \ml \left[\prod_{(i,j) \in S} z_{i,j} x_i x_j \right]\]

Suppose we assume that $F(\x,\z)$ has a polynomial-size $\IPSLINp$ refutation. This implies that there is a multilinear circuit of polynomial size computing $G(\x,\z)$. We further isolate the degree ${n/2 \choose 2}$ component in $\z$ variables by interpolating it out and further degree $n$ component in $\x$ variables by another interpolation. 
It is easy to see that the polynomial this computes equals $C_{n/2}(\x,\z)$ (up to scaling by a coefficient). Assuming $\mathsf{VP} \neq \mathsf{VNP}$, this gives a contradiction. 
\end{proof}

\end{document}